\documentclass[sigconf]{acmart}
\usepackage{enumitem}
\usepackage{threeparttable}
\AtBeginDocument{%
  \providecommand\BibTeX{{%
    \normalfont B\kern-0.5em{\scshape i\kern-0.25em b}\kern-0.8em\TeX}}}

\copyrightyear{2022}
\acmYear{2022}
\setcopyright{rightsretained}
\acmConference[SIGIR '22]{Proceedings of the 45th International ACM
SIGIR Conference on Research and Development in Information
Retrieval}{July 11--15, 2022}{Madrid, Spain}
\acmBooktitle{Proceedings of the 45th International ACM SIGIR
Conference on Research and Development in Information Retrieval (SIGIR
'22), July 11--15, 2022, Madrid, Spain}\acmDOI{10.1145/3477495.3532000}
\acmISBN{978-1-4503-8732-3/22/07}
\begin{document}
\fancyhead{}
\title{INMO: A Model-Agnostic and Scalable Module for Inductive Collaborative Filtering}

%\author{Yunfan Wu$^{1,3}$, Qi Cao$^{1,*}$, Huawei Shen$^{1,3,*}$, Shuchang Tao$^{1,3}$, Xueqi Cheng$^{2,3}$}
%\affiliation{
%  \institution{$^{1}$Data Intelligence System Research Center,  \\
%  Institute of Computing Technology, Chinese Academy of Sciences, Beijing, China\\
%  $^{2}$CAS Key Laboratory of Network Data Science and Technology,}
%  \country{Institute of Computing Technology, Chinese Academy of Sciences, Beijing, China\\
%  $^{3}$University of Chinese Academy of Sciences, Beijing, China
%}
%}
%\email{{wuyunfan19b,caoqi,shenhuawei,taoshuchang18z,cxq}@ict.ac.cn}
\author{Yunfan Wu}
\affiliation{%
  \institution{Data Intelligence System Research Center, Institute of Computing}
  \country{Technology, CAS}
}
\affiliation{%
  \institution{University of Chinese Academy of}
  \country{Sciences, Beijing, China}
}
\email{wuyunfan19b@ict.ac.cn}

\author{Qi Cao}
\authornotemark[1]
\affiliation{%
  \institution{Data Intelligence System Research Center, Institute of Computing}
  \country{Technology, CAS}
}
\email{caoqi@ict.ac.cn}

\author{Huawei Shen}
\authornotemark[1]
\affiliation{%
  \institution{Data Intelligence System Research Center, Institute of Computing}
  \country{Technology, CAS}
}
\affiliation{%
  \institution{University of Chinese Academy of}
  \country{Sciences, Beijing, China}
}
\email{shenhuawei@ict.ac.cn}

\author{Shuchang Tao}
\affiliation{%
  \institution{Data Intelligence System Research Center, Institute of Computing}
  \country{Technology, CAS}
}
\affiliation{%
  \institution{University of Chinese Academy of}
  \country{Sciences, Beijing, China}
}
\email{taoshuchang18z@ict.ac.cn}

\author{Xueqi Cheng}
\affiliation{%
  \institution{CAS Key Lab of Network Data Science and Technology, Institute of}
  \country{Computing Technology, CAS}
}
\affiliation{%
  \institution{University of Chinese Academy of}
  \country{Sciences, Beijing, China}
}
\email{cxq@ict.ac.cn}
\renewcommand{\shortauthors}{Wu et al.}

\begin{abstract}
  Collaborative filtering is one of the most common scenarios and 
  popular research topics in recommender systems.
  Among existing methods, latent factor models, i.e., learning a specific embedding for each user/item by 
  reconstructing the observed interaction matrix, have shown excellent performances.
  However, such user-specific and item-specific embeddings are intrinsically transductive, making it difficult to 
  deal with new users and new items unseen during training.
  Besides, the number of model parameters heavily depends on the number of all users and items, 
  restricting its scalability to real-world applications.
  To solve the above challenges, in this paper, we propose a novel model-agnostic and scalable 
  \textbf{In}ductive Embedding \textbf{Mo}dule for collaborative filtering, namely INMO. 
  INMO generates the inductive embeddings for users (items) by characterizing their 
  interactions with some template items (template users), instead of employing an embedding lookup table. 
  Under the theoretical analysis, we further propose an effective indicator for the selection of template users/items.
  Our proposed INMO can be attached to existing latent factor models as a pre-module, inheriting the expressiveness of backbone models, 
  while bringing the inductive ability and reducing model parameters. 
  % In practice, with the effective selection indicator, we could select only $30\%$ of all users/items as the representative users/items, 
  % which retains a similar recommendation performance but requires much fewer model parameters.
  We validate the generality of INMO by attaching it to both Matrix Factorization (MF) 
  and LightGCN, which are two representative latent factor models for collaborative filtering. 
  Extensive experiments on three public benchmarks demonstrate the effectiveness and efficiency of INMO 
  in both transductive and inductive recommendation scenarios.
  %Our implementations are available here in PyTorch\footnote{\url{https://anonymous.4open.science/r/INMO-LGCN_cf/}}.
   \let\thefootnote\relax\footnotetext{*Corresponding Authors.}
\end{abstract}

\begin{CCSXML}
  <ccs2012>
     <concept>
         <concept_id>10002951.10003317.10003347.10003350</concept_id>
         <concept_desc>Information systems~Recommender systems</concept_desc>
         <concept_significance>500</concept_significance>
         </concept>
   </ccs2012>
\end{CCSXML}
  
\ccsdesc[500]{Information systems~Recommender systems}

\keywords{Recommender System, Collaborative Filtering, Inductive Embedding, Latent Factor Model}


\maketitle
%{\fontsize{8pt}{8pt} \selectfont
%\textbf{ACM Reference Format:}\\
%Yunfan Wu, Qi Cao, Huawei Shen, Shuchang Tao, Xueqi Cheng. 2022. 
%INMO: A Model-Agnostic and Scalable Module for Inductive Collaborative Filtering.
%In {\it Proceedings of the 45th International ACM SIGIR Conference on Research and Development in Information Retrieval
%(SIGIR'22), July 11–15, 2022, Madrid, Spain.} ACM, New York, NY, USA, 11 pages. https://doi.org/10.1145/3477495.3\\532000}

\section{Introduction}
Recommender systems are prevalently deployed in real-world applications to provide personalized recommendation services, 
helping people out of the dilemma of information overload \cite{goldberg1992using, covington2016deep, ying2018graph}. 
Among various recommendation tasks, collaborative filtering (CF) is one of the most simple and widely adopted scenarios, 
which has attracted extensive research attention for more than two decades \cite{su2009survey}.
%While early literature on collaborative filtering has mainly focused on the explicit rating prediction \cite{koren2009matrix}, recent works care more about the top-k recommendation based on implicit interactions \cite{he2017neural, Wang2020SIGIR, he2020lightgcn}, which is also the focus of this paper.
\begin{figure}[t]
  \centering
  \includegraphics[width=\linewidth]{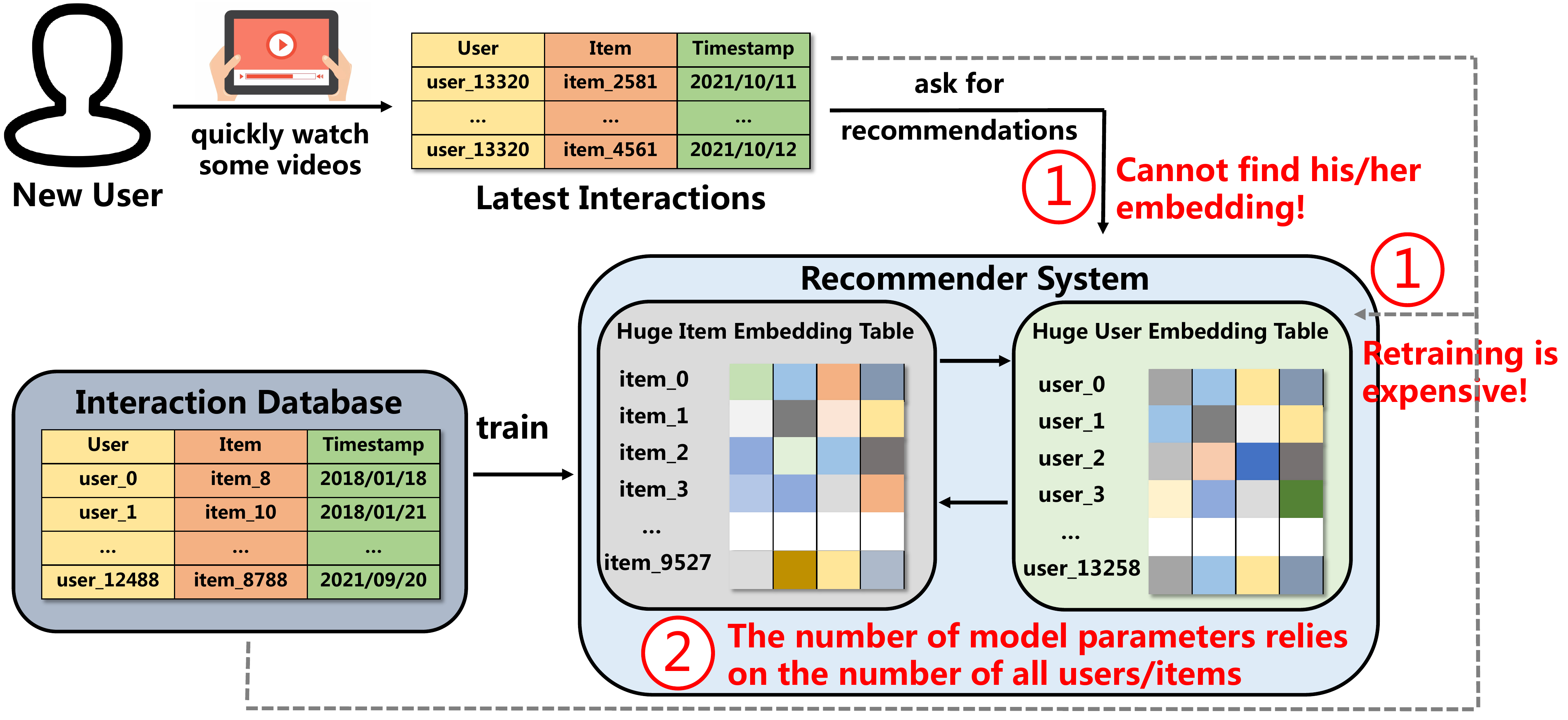}
  \caption{Limitations of existing latent factor models for collaborative filtering.}
  \Description{A newly registered user cannot obtain his/her recommendation list from a trained recommender.}
  \label{intro}
  \vspace{-0.5cm} 
\end{figure}

Among existing methods, latent factor models have been the state-of-the-art in CF for over a decade 
\cite{hofmann1999latent, koren2008factorization}.
Since high-quality side information is not always available~\cite{pmlrwu21j}, 
the most general paradigm of latent factor models is to project the ID of a user (item) to a specific learnable embedding, 
and then predict user-item interactions based on these embeddings~\cite{koren2009matrix}. 
Recently, with the success of deep learning, researchers further improve latent factor models from two aspects, 
i.e., representation generation \cite{wang2015collaborative, fan2019graph, wang2019kgat, BergKDD18} and interaction modeling \cite{he2017neural, he2018outer}.
The former line generates more informative representations based on initial user/item embeddings, 
e.g., utilizing graph neural networks to capture the high-order proximity~\cite{BergKDD18}.
The latter line devotes to enhancing the interaction modeling between users and items with powerful neural networks 
instead of a simple inner production, e.g., employing multi-layer perceptrons to learn the complex interaction function~\cite{he2017neural}.

However, all the above latent factor models rely on the user-specific (or item-specific) embedding learning, 
which have two critical limitations when dealing with real-world recommendation scenarios (shown in Figure \ref{intro}).
First, such methods are intrinsically transductive, making them hard to deal with continuous new users and new items.
For example, there may be newly registered Youtube users who quickly watch some preferred videos \cite{MuhanICLR20} and it is a practical need to make
personal recommendations for them according to such latest behaviors.
Unfortunately, with the embedding lookup paradigm, these methods cannot find the embeddings for new users or new items 
which are unseen during training, while the cost of retraining or incremental learning for new users/items 
is generally expensive~\cite{rendle08online, incre2016yu}.
Moreover, with the embedding lookup table, the number of model parameters heavily depends on the number of all users and items, 
restricting their scalability to real-world applications with hundreds of millions of users and items.

Recently, a few works have also noticed the limitations of the transductive nature of existing latent factor models and attempted to 
propose inductive collaborative filtering methods without side information \cite{JasonICML18, MuhanICLR20, pmlrwu21j, shen2021inductive, ragesh2021user}.
\citet{ying2018graph}. 
However, they either need an expensive computational cost~\cite{JasonICML18, MuhanICLR20, pmlrwu21j} or have a limited recommendation accuracy 
for achieving the inductiveness \cite{shen2021inductive}.
For example, \citet{pmlrwu21j} present a two-stage framework (IDCF) to learn the 
latent relational graph among existing users and new users, which requires a quadratic complexity.
\citet{shen2021inductive} use a global average embedding for new users, 
losing the personalization of recommendations for achieving the inductiveness.
To sum up, there still lacks a both efficient and effective inductive collaborative filtering method.

In this paper, we formally define the inductive recommendation scenarios and 
address the aforementioned problems by proposing a novel \textbf{In}ductive Embedding \textbf{Mo}dule (INMO) for collaborative filtering. 
Specifically, INMO generates the inductive embedding of a user by considering its past interactions 
with a set of template items (vice versa), instead of learning a specific embedding for each user and item. 
As long as a new user (item) has interacted with the preselected template items (users), INMO could generate 
an informative embedding for the new user (item).
Besides, the number of parameters in INMO only depends on the number of template users and template items, 
which is adjustable according to available computing resources, contributing to its better scalability to real-world applications.
Under the theoretical analysis, we further propose an effective indicator for the selection of template users/items, making it possible for INMO to achieve competitive recommendation performances with much fewer model parameters.

Remarkably, our proposed INMO is model-agnostic, which can be easily attached to all existing latent factor CF methods 
as a pre-module, inheriting the expressiveness of backbone models, while bringing the inductive ability and reducing model parameters. 
We experiment INMO with the representative Matrix Factorization (MF) \cite{koren2009matrix} 
and the state-of-the-art LightGCN \cite{he2020lightgcn} to show its generality.
Extensive experiments conducted on three public benchmark datasets, across both transductive and inductive recommendation scenarios,
demonstrate the effectiveness of our proposed INMO.

% To summarize, our work makes the following main contributions:
To summarize, our contributions are as follows:
\begin{itemize}[leftmargin=*]
  \item We formally define two inductive recommendation scenarios to progressively evaluate the inductive ability for CF methods.
  \item We propose a novel \textbf{In}ductive Embedding \textbf{Mo}dule (INMO), which is applicable to existing latent factor models, 
  bringing the inductive ability and reducing model parameters.
  \item Extensive experiments conducted on three real-world datasets across both transductive and inductive recommendation scenarios 
  demonstrate the effectiveness and generality of INMO.
\end{itemize}

\section{Related Work}
This section briefly reviews existing works on latent factor collaborative filtering methods and discusses several recent inductive
recommenders, which are the most relevant to this work. 

\subsection{Latent Factor CF Methods}
Latent factor models have been the state-of-the-art in collaborative filtering for over a decade \cite{koren2008factorization}.
These models learn vectorized embeddings for users and items by reconstructing the original user-item interaction data 
\cite{hofmann1999latent, koren2008factorization}.
From a systematic view, most existing latent factor models have two key components, 
i.e., representation generation and interaction modeling.

To improve the representation generation, many methods have been proposed to incorporate the external side information, 
like item attributes~\cite{wang2015collaborative, chen2017attentive}, social networks~\cite{ma2008sorec, fan2019graph}, 
knowledge graphs~\cite{palumbo2017entity2rec, wang2019kgat}, etc.
However, high-quality side information is not always available in real-world applications.
Recently, with the success of graph neural networks (GNNs) in various fields~\cite{cao2020pop, xu2019graph}, they have also been introduced into the collaborative filtering task \cite{BergKDD18, zheng2018spectral, ying2018graph, wang2019neural, he2020lightgcn}.
These GNN-based methods could generate more comprehensive representations for users and items, 
capturing their high-order relationships by iteratively aggregating neighbor information in the user-item interaction graph. 
Among them, LightGCN \cite{he2020lightgcn} is a light but effective CF model, achieving the state-of-the-art performance.

As for interaction modeling, while the inner product is a commonly adopted and efficient choice~\cite{koren2009matrix, rendle2020neural}, the
linearity makes it insufficient to reveal the complex and nonlinear interactions between users and items \cite{he2017neural}.
Some variants of MF \cite{koren2009matrix, koren2008factorization} add several bias terms to the inner product 
for better preference modeling.
\citet{hsieh2017collaborative} employ the Euclidean distance instead of the inner product to estimate the similarities between users and items.
\citet{he2017neural} propose NeuMF, introducing a multi-layer perceptron (MLP) to learn the highly non-linear user-item interaction function.
Following NeuMF, \citet{he2018outer} use a convolutional neural network, modeling high-order correlations between representations.

However, despite the great success of latent factor models in CF,
%most existing latent factor models transforms the one-hot index to 
%vector representations in the representation generation, while the interaction modeling serves as a similarity function on the representations.
%However, 
most existing methods rely on the embedding lookup table for user/item-specific embeddings, 
making these models intrinsically transductive and limiting their scalability.

\begin{figure}[t]
  \centering
  \includegraphics[width=\linewidth]{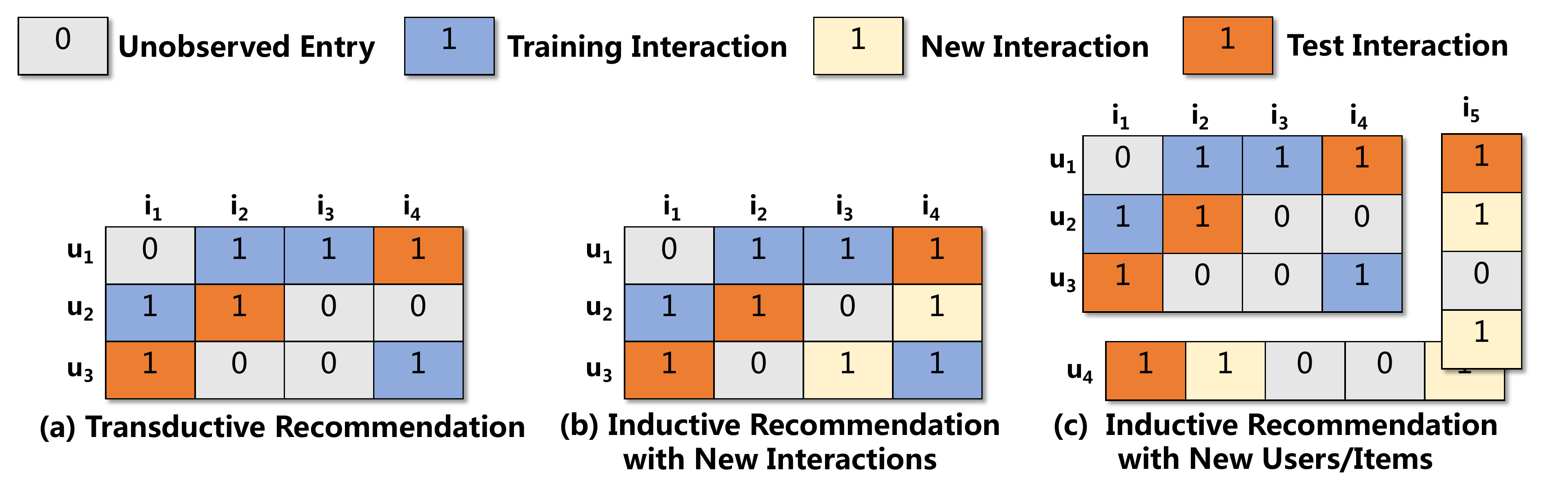}
  \caption{Traditional transductive recommendation scenario and two inductive scenarios proposed in this paper.}
  \Description{Three example interaction matrices denoting the recommendation scenarios.}
  \label{settings}
\end{figure}

\subsection{Inductive CF Methods}
In real-world applications, there are always new interactions, as well as newly registered users and items.
%existing users continuously make new interactions with items, and 
Practical recommender systems need to be periodically retrained to refresh the models with the new data \cite{zhang2020retrain}.
Though some works have discussed incremental and online recommenders to incorporate new users/items \cite{rendle08online, incre2016yu}, 
they still bring an additional training cost or incremental error.
As a result, it is a valuable property for the recommenders to possess the inductive ability, i.e., 
directly characterizing additional new interactions and new users/items.

It is worth noting that, the inductive scenario is different from the \emph{cold-start problem}. 
The cold-start methods focus on improving recommendations for users with few or no interactions \cite{lee2019melu, Lin00PCW21}. 
In contrast, the inductive scenario requires the recommenders to incorporate the new data and update their predictions without the need of retraining. 
An inductive recommender could predict the changeful preference of an existing user or a completely new user, according to its latest behaviors. 
In a word, the inductiveness discussed in this paper concentrates on the dynamic modeling capacity, instead of improving the experiences for long-tailed users.

Recently, a few works have discussed recommendations in the inductive scenarios
\cite{ying2018graph, JasonICML18, MuhanICLR20, pmlrwu21j, shen2021inductive, ragesh2021user}.
\citet{ying2018graph} propose PinSage, adopting an inductive variant of GNNs to make recommendations for 
an online content discovery service Pinterest. 
It leverages the visual and annotation features of pins as inputs to achieve inductiveness.
However, high-quality side information is inadequate in many situations \cite{MuhanICLR20}.
\citet{JasonICML18} study the matrix completion in an inductive scenario, proposing an exchangeable matrix layer to 
do the message passing between interactions and use the numerical ratings as input features.
But it is neither time-efficient nor can be applied to the implicit feedback data without real-valued ratings.
\citet{MuhanICLR20} (IGMC) consider the rating prediction as a graph-level regression task.
They define some heuristic node features and predict the rating of a user-item pair by learning the local graph pattern.
Despite its inductive ability, IGMC has to do subgraph extraction and graph regression for every user-item pair independently, 
resulting in an unaffordable time for the top-k recommendation task.
Besides, \citet{pmlrwu21j} present a two-stage framework to estimate the relations from key users to query users, 
which takes a quadratic complexity.
A very recent work IMC-GAE \cite{shen2021inductive} employs a postprocessing method while losing the personalization of new users.

To conclude, existing inductive CF methods are either time-consuming or have limited recommendation accuracy. 
There still lacks a both efficient and effective inductive method for the collaborative filtering task.

%\vspace{-5pt}
\section{Preliminaries}
In this section, with a brief review of the commonly studied transductive recommendation task, 
we first propose and formalize two inductive recommendation tasks.
%and provide formal definitions for them. 
Afterward, we introduce two representative latent factor methods for CF, i.e.,  
MF \cite{koren2009matrix} and LightGCN \cite{he2020lightgcn}, which are the backbone models to 
apply INMO for experiments.
%%\vspace{-5pt}

\subsection{Transductive and Inductive CF Scenarios}
\label{sec:settings}
Existing researches generally evaluate the performance of recommender systems in a transductive scenario. 
Specifically, they select a part of observed interactions from each user to serve as the training data, 
and treat the remaining interactions as the test data (Figure \ref{settings}(a)).
All users and items are assumed to have been seen during training. 
However, in a real-world recommendation service, there are always newly registered users and newly created items, 
as well as new interactions between existing users and items.
Such new users/items and new interactions emerge continuously and have not been seen during training. 
%For example, a newly registered Youtube user may quickly watch some videos as his or her latest interactions. 
As a result, it is necessary to evaluate recommender systems from an inductive view, 
i.e., the ability to make recommendations for new users/items with new interactions that are unseen during the training phase.
\begin{figure*}[ht]
  \centering
  \includegraphics[width=0.9\linewidth]{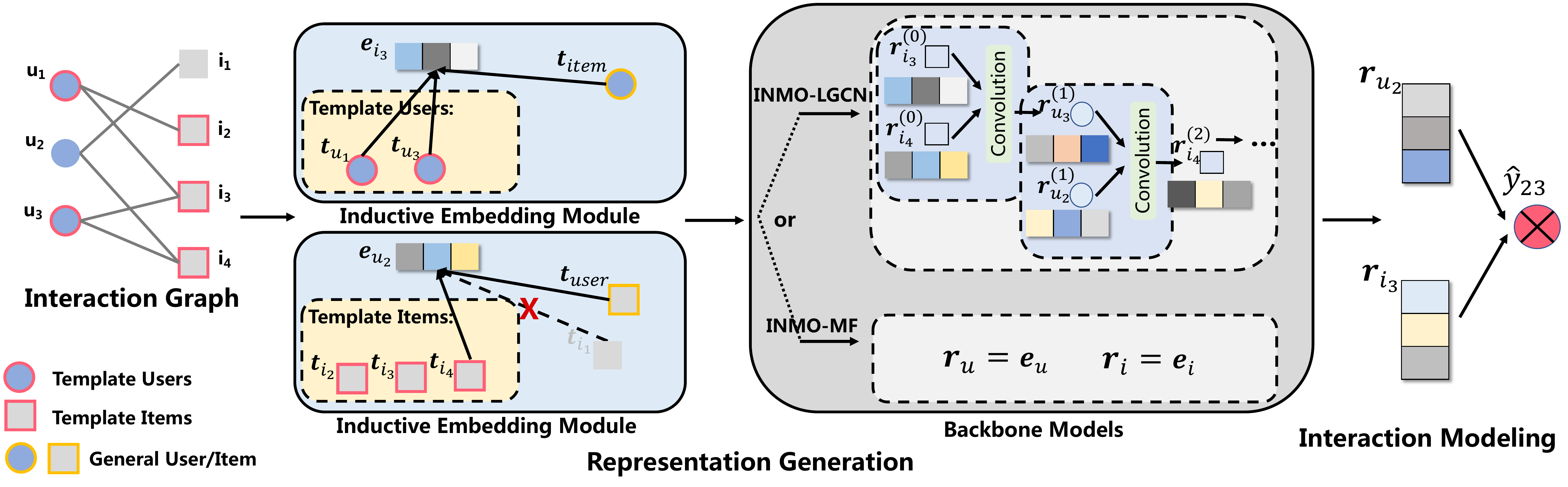}
  \caption{The overall architecture of an inductive CF recommender, which consists of two key components, 
  i.e., representation generation and interaction modeling. Here, we use INMO implemented with MF (INMO-MF) and LightGCN (INMO-LGCN) as two specific instances. 
  The template users (template items) are some carefully selected users (items) to serve as the bases to characterize every item (user) in the embedding space.}
  \Description{An example user-item interaction graph followed by two components of an inductive latent factor recommendation model.}
  \label{model}
\end{figure*}

In this work, we propose two specific scenario settings of inductive recommendations,
which could better evaluate the dynamic modeling capability of recommender systems. 
The first is the \textbf{inductive recommendation with new interactions}, 
where some additional new interactions between existing users and items are observed after training, see Figure \ref{settings}(b).
Formally, $\mathcal{N}_\mathsf{u}$ denotes the set of interacted items of user $\mathsf{u}$, 
and $\mathcal{N}_\mathsf{i}$ denotes the set of interacted users of item $\mathsf{i}$.
In the test phase, the extended interaction sets 
$\mathcal{N}_\mathsf{u}^{test-ob} = \mathcal{N}_\mathsf{u}^{train} \cup  \mathcal{N}_\mathsf{u}^{new}, 
\mathcal{N}_\mathsf{i}^{test-ob} = \mathcal{N}_\mathsf{i}^{train} \cup \mathcal{N}_\mathsf{i}^{new}$ are observed for each user and item.
In this scenario, it requires the recommenders to flexibly incorporate new interactions and predict the updated user preference without retraining.

The second scenario is the \textbf{inductive recommendation with new users/items}, where some new users $U^{new}$ and new items $I^{new}$ are 
created after the training phase.
This scenario is different from the extreme cold start problem \cite{MuhanICLR20}, since 
the new users and new items should have at least some observed interactions at the test phase, 
according to which the models could make accurate recommendations (Figure \ref{settings}(c)).
Formally, 

\begin{footnotesize}
\begin{align*}
U^{test} = U^{train} \cup U^{new}&,\ I^{test} = I^{train} \cup I^{new}; \\
\forall \mathsf{u} \in U^{train}, \ \mathcal{N}_\mathsf{u}^{test-ob} = \mathcal{N}_\mathsf{u}^{train} \cup \mathcal{N}_\mathsf{u}^{new}&;\ 
\forall \mathsf{u} \in U^{new}, \ \mathcal{N}_\mathsf{u}^{test-ob} = \mathcal{N}_\mathsf{u}^{new}; \\
\forall \mathsf{i} \in I^{train}, \ \mathcal{N}_\mathsf{i}^{test-ob} = \mathcal{N}_\mathsf{i}^{train} \cup \mathcal{N}_\mathsf{i}^{new}&;\ 
\forall \mathsf{i} \in I^{new}, \ \mathcal{N}_\mathsf{i}^{test-ob} = \mathcal{N}_\mathsf{i}^{new}.
\end{align*}
\end{footnotesize}
\noindent The inductive evaluation in this scenario expects the recommenders to accurately recommend items for new users
and recommend new items to users, which is a prevalent need in practice. 

\subsection{MF and LightGCN}
Let $U = \{\mathsf{u}_1, \mathsf{u}_2, \cdots, \mathsf{u}_n\}$ and $I = \{\mathsf{i}_1, \mathsf{i}_2, \cdots, \mathsf{i}_m\}$ 
denote the set of users and items in a recommendation system.
Matrix Factorization (MF)~\cite{koren2009matrix} is the most basic latent factor model, which directly obtains the final 
representations of users and items from an embedding lookup table,
\begin{equation}
  \label{eq:mf1}
  \boldsymbol{r}_\mathsf{u}=\boldsymbol{e}_\mathsf{u}, \ \boldsymbol{r}_\mathsf{i}=\boldsymbol{e}_\mathsf{i}. 
\end{equation}
Here, $\boldsymbol{e}_{\mathsf{u}}, \boldsymbol{e}_{\mathsf{i}} \in \mathbb{R}^d$ are the embeddings of user $\mathsf{u}$ and item $\mathsf{i}$ 
through an embedding lookup table, and $\boldsymbol{r}_{\mathsf{u}}, \boldsymbol{r}_{\mathsf{i}}$ are their final representations.

LightGCN \cite{he2020lightgcn} is a state-of-the-art CF recommender. After obtaining the initial representations 
$\boldsymbol{r}_\mathsf{u}^{(0)}=\boldsymbol{e}_\mathsf{u},\ \boldsymbol{r}_\mathsf{i}^{(0)}=\boldsymbol{e}_\mathsf{i}$, 
it leverages a linear GNN to refine representations by iteratively aggregating neighbor information in the user-item interaction graph, i.e.,
\begin{equation}  \label{eq:aggregation}
   \boldsymbol{r}_{\mathsf{u}}^{(l+1)} = AGG(\{\boldsymbol{r}_{\mathsf{i}}^{(l)}: \mathsf{i} \in \mathcal{N}_\mathsf{u}\})~,
 \end{equation}
\noindent where $\boldsymbol{r}_\mathsf{u}^{(l)}$ is the representation of user $\mathsf{u}$ at the $l$-th layer.
%$\mathcal{N}_\mathsf{u}$ is the set of interacted items of user $\mathsf{u}$.
The final representations are obtained by an average of all layers, i.e,
$\boldsymbol{r}_{\mathsf{u}} = \frac{1}{K+1} \sum_{l=0}^{K}\boldsymbol{r}_{\mathsf{u}}^{(l)},\ 
\boldsymbol{r}_{\mathsf{i}} = \frac{1}{K+1} \sum_{l=0}^{K}\boldsymbol{r}_{\mathsf{i}}^{(l)}$.

Both of MF and LightGCN employ a simple inner product to predict the preference score $\hat{y}_{ij}$ between user $\mathsf{u}_i$ and item $\mathsf{i}_j$, 
based on their final representations, i.e., $\hat{y}_{ij} = \boldsymbol{r}_{\mathsf{u}_i}^T \boldsymbol{r}_{\mathsf{i}_j}$.

To attach our proposed INMO to MF or LightGCN, we 
just replace the transductive embedding lookup table with INMO which can generate the embeddings $\boldsymbol{e}_{\mathsf{u}}$ and $\boldsymbol{e}_{\mathsf{i}}$ in an inductive manner.

\section{Methodology}
This section first introduces our intuitions behind the design of INMO, which provides both theoretical and empirical analyses.
Then, we propose the novel \textbf{In}ductive Embedding \textbf{Mo}dule, to make recommendations in the inductive scenarios with an adjustable number of parameters.
Lastly, an additional self-enhanced loss and two training techniques are presented for model optimization.
The overall architecture of an inductive CF recommender with INMO is illustrated in Figure \ref{model}.

\subsection{Theoretical Analysis}
\label{sec:th}
In the classic collaborative filtering setting without any side information, most existing latent factor models leverage an 
embedding lookup table, mapping the one-hot index of a user/item to an embedding vector. 
However, such an embedding lookup table is intrinsically transductive and brings the scalability difficulty.
In this work, we aim to propose a scalable inductive embedding module, which could inductively generate the
embeddings for new users and new items.

Before diving into our proposed method, let us first review the fundamental assumption of CF, i.e., 
if two users have similar past interactions with items, they will act on other items similarly in the future~\cite{su2009survey, ning2011SLIM, he2018nais, sarwar2001item}.
Based on such assumption, in this paper, we propose to design the inductive embedding module for users (items) 
by considering their past interactions with some carefully selected template items (template users).
Such template items (template users) serve as a set of bases, the combination of which could represent different user preferences (item characteristics).

Let $U_{tem}$, $I_{tem}$ denote the sets of template users and template items, 
and $\boldsymbol{T}_{\mathsf{u}} \in \mathbb{R}^{n_t \times d},\boldsymbol{T}_{\mathsf{i}} \in \mathbb{R}^{m_t \times d}$ 
denote the template user vectors and template item vectors, where $n_t=|U_{tem}|, m_t=|I_{tem}|$.
We expect to design two inductive functions $f_{\mathsf{u}}$ and $f_{\mathsf{i}}$, to generate the embeddings for users and items 
according to their interactions with the template items/users, i.e., 
$\boldsymbol{e}_{\mathsf{u}} = f_{\mathsf{u}}(\mathcal{N}_{\mathsf{u}} \cap I_{tem}, \boldsymbol{T}_{\mathsf{i}})$ and 
$\boldsymbol{e}_{\mathsf{i}} = f_{\mathsf{i}}(\mathcal{N}_{\mathsf{i}} \cap U_{tem}, \boldsymbol{T}_{\mathsf{u}})$.
Considering the recent finding that the nonlinear transformation adopted by neural networks is burdensome for the CF task 
\cite{he2020lightgcn,simplify19wu,Wang2020SIGIR},
in this paper, we present a really simple yet theoretical effective design of $f_{\mathsf{u}}$ and $f_{\mathsf{i}}$, 
which brings little optimization difficulty and has sufficient expressiveness from 
$\boldsymbol{T}_{\mathsf{u}}$ and $\boldsymbol{T}_{\mathsf{i}}$. Formally, we define
\begin{equation} \label{eq:initialembedding}
   \boldsymbol{e}_{\mathsf{u}} = \sum_{\mathsf{i} \in \mathcal{N}_{\mathsf{u}} \cap I_{tem}}  \boldsymbol{t}_\mathsf{i}, \quad
   \boldsymbol{e}_{\mathsf{i}} = \sum_{\mathsf{u} \in \mathcal{N}_{\mathsf{i}} \cap U_{tem}}  \boldsymbol{t}_\mathsf{u},
\end{equation}
where $\boldsymbol{t}_\mathsf{i}$ is the template item vector of $\mathsf{i}$. Through such inductive functions, we can generate the inductive embeddings for new users and new items which are unseen during the training phase.

We first theoretically prove the expressiveness of INMO in Eq.~(\ref{eq:initialembedding}) based on a representative latent factor model 
and then present a both theoretically and empirically effective indicator to determine the set of template users/items. 
\begin{theorem}
  \label{th:1}
	Assuming the original MF can achieve a matrix factorization error $\epsilon$ on the 
  interaction matrix $\boldsymbol{Y}$, then there exists a solution for INMO-MF 
  such that its error is less than or equal to $\epsilon$, when we take all users/items as the template users/items.
\end{theorem}
\begin{proof}
Here, we theoretically prove that, when we take all users/items as the template users/items, 
with the embeddings $\boldsymbol{e}_{\mathsf{u}}$ and $\boldsymbol{e}_{\mathsf{i}}$ generated 
from  Eq.~(\ref{eq:initialembedding}), the performance of INMO-MF would not be worse than the original MF.

Essentially, matrix factorization aims to do a low rank approximation on the interaction matrix, i.e., minimizing the difference between $\boldsymbol{Y}$ and $\boldsymbol{E}_{\mathsf{u}} \boldsymbol{E}_{\mathsf{i}}^T$, where $\boldsymbol{Y} \in \{0, 1\}^{n \times m}$ 
is the observed interaction matrix between users and items. 
According to the Eckart-Young theorem, we can get 
$\boldsymbol{Y}=\boldsymbol{U}_{n \times p}\boldsymbol{S}_{p \times p}\boldsymbol{V}^T_{p \times m}=
\boldsymbol{U}^d_{n \times d}\boldsymbol{S}^d_{d \times d}(\boldsymbol{V}^d)^T_{d \times m}+
\boldsymbol{U}^{\epsilon}_{n \times (p-d)}\boldsymbol{S}^{\epsilon}_{(p-d) \times (p-d)}(\boldsymbol{V}^{\epsilon})^T_{(p-d) \times m}$, 
where $p=min(n,m)$. $\boldsymbol{S}$ is a diagonal matrix whose elements are singular values of $\boldsymbol{Y}$, 
and $\boldsymbol{U}, \boldsymbol{V}$ are column orthogonal matrices. 
$\boldsymbol{U}^d\boldsymbol{S}^d(\boldsymbol{V}^d)^T$ with $d$ largest singular values in $\boldsymbol{S}^d$, 
is the closest rank-$d$ matrix to $\boldsymbol{Y}$ in both 
Frobenius norm and spectral norm, denoted as $|Y-U^d S^d(V^d)^T|_F =\epsilon_{min}$, where $\epsilon_{min}=min_{rank(\hat{Y})=d}|Y-\hat{Y}|_F$ and $\hat{Y}=\boldsymbol{E}_{\mathsf{u}} \boldsymbol{E}_{\mathsf{i}}^T$ is the low rank approximation of MF.

Next, we show that with 
$\boldsymbol{E}_{\mathsf{u}}=\boldsymbol{Y}\boldsymbol{T}_{\mathsf{i}}, 
\boldsymbol{E}_{\mathsf{i}}=\boldsymbol{Y}^T\boldsymbol{T}_{\mathsf{u}}$ (the matrix form of Eq.~(\ref{eq:initialembedding})), 
INMO could learn the closest rank-$d$ solution of $\boldsymbol{Y}$. 
In other words, with the same embedding dimension, INMO-MF can achieve the error $\epsilon_{min}$ which is the minimum 
possible error obtained by any solutions of MF.
Specifically, there exists a solution for INMO-MF as $\boldsymbol{T}_{\mathsf{u}}=\boldsymbol{U}^d(\boldsymbol{S}^d)^{-1}, 
\boldsymbol{T}_{\mathsf{i}}=\boldsymbol{V}^d$, that has
%\begin{small}
%  \begin{equation}
%    \begin{split}
%        |\boldsymbol{Y} - \boldsymbol{E}_{\mathsf{u}}\boldsymbol{E}_{\mathsf{i}}^T|_F
%      = &|\boldsymbol{Y} - \boldsymbol{Y}\boldsymbol{V}^d(\boldsymbol{S}^d)^{-1}(\boldsymbol{U}^d)^T\boldsymbol{Y}|_F \\
%      = &|\boldsymbol{Y} - (\boldsymbol{U}^d\boldsymbol{S}^d(\boldsymbol{V}^d)^T + \boldsymbol{U}^{\epsilon}\boldsymbol{S}^{\epsilon}(\boldsymbol{V}^{\epsilon})^T)
%      \boldsymbol{V}^d(\boldsymbol{S}^d)^{-1}(\boldsymbol{U}^d)^T\boldsymbol{Y}|_F \\
%      = &|\boldsymbol{Y} - \boldsymbol{U}^d(\boldsymbol{U}^d)^T\boldsymbol{Y}|_F 
%      = |\boldsymbol{Y} - \boldsymbol{U}^d\boldsymbol{S}^d(\boldsymbol{V}^d)^T|_F = \epsilon_{min}.
%    \end{split}
%  \end{equation}
%\end{small}

\begin{small}
  \begin{equation}
    \begin{split}
        |\boldsymbol{Y} - \boldsymbol{E}_{\mathsf{u}}\boldsymbol{E}_{\mathsf{i}}^T|_F
      = &|\boldsymbol{Y} - \boldsymbol{Y}\boldsymbol{V}^d(\boldsymbol{S}^d)^{-1}(\boldsymbol{U}^d)^T\boldsymbol{Y}|_F \\
      = &|\boldsymbol{Y} - \boldsymbol{U}^d(\boldsymbol{U}^d)^T\boldsymbol{Y}|_F 
      = |\boldsymbol{Y} - \boldsymbol{U}^d\boldsymbol{S}^d(\boldsymbol{V}^d)^T|_F = \epsilon_{min}.
    \end{split}
  \end{equation}
\end{small}

%Note that we have $((\boldsymbol{S}^d)^{-1})^T=(\boldsymbol{S}^d)^{-1},
%(\boldsymbol{U}^d)^T\boldsymbol{U}^d=(\boldsymbol{V}^d)^T\boldsymbol{V}^d=\boldsymbol{E}_{k \times k},
%(\boldsymbol{V}^{\epsilon})^T\boldsymbol{V}^d=\boldsymbol{0}_{(p-d) \times d},
%(\boldsymbol{U}^{d})^T\boldsymbol{U}^{\epsilon}=\boldsymbol{0}_{d \times (p-d)}$.
\end{proof}
The above proof validates that, our \textbf{INMO-MF has at least the same expressiveness as the original MF} %in the transductive scenario, 
while being capable to do inductive recommendations. 
A similar proof could be deduced for INMO-LGCN, which is a linear model as well.

The next important question is how to carefully select the template users/items in order to reduce model parameters while 
avoiding the additional error as much as possible.
When we only select a part of users/items as the template ones, the INMO in Eq.~(\ref{eq:initialembedding}) can be written 
as $\boldsymbol{E}_{\mathsf{u}}=\boldsymbol{Y}_{n \times m}(\boldsymbol{C}_\mathsf{i})_{m \times m_t}(\boldsymbol{T}_{\mathsf{i}})_{m_t \times d}$. 
%It is natural to change $\boldsymbol{E}_{\mathsf{u}}=\boldsymbol{Y}_{n \times m}(\boldsymbol{T}_{\mathsf{i}})_{m \times d}$ to
%$\boldsymbol{E}_{\mathsf{u}}=(\boldsymbol{Y}_\mathsf{i})_{n \times m_t}(\boldsymbol{T}_{\mathsf{i}})_{m_t \times d}$, 
% where $m_t$ is the number of template items.
Each column of $\boldsymbol{C}_\mathsf{i}$ is an one-hot vector, and each row of $\boldsymbol{C}_\mathsf{i}$ has at most one non-zero entry.
$(\boldsymbol{C}_\mathsf{i})_{i,j}=1$ means that the $i$-th item $\mathsf{i}_i$ is selected as the $j$-th template item $(\mathsf{i}_{tem})_j$.
Similarly, $\boldsymbol{E}_{\mathsf{i}}=\boldsymbol{Y}^T_{m \times n}(\boldsymbol{C}_\mathsf{u})_{n \times n_t}(\boldsymbol{T}_{\mathsf{u}})_{n_t \times d}$.
The number of model parameters in INMO is now $(n_t+m_t)d$, which is much smaller than the original embedding table with $(n+m)d$ parameters.
\begin{theorem}
  \label{th:2}
  When selecting those users $\mathsf{u}_j$ with the largest 
  $|\boldsymbol{s}^{\mathsf{u}}_j|_2^2 \sum_{\mathsf{i} \in \mathcal{N}_{\mathsf{u}_j}} |\mathcal{N}_{\mathsf{i}}|$ as the template users, 
  INMO minimizes an upper bound of the additional error caused by ignoring non-template users.
\end{theorem}
\begin{proof}
Similar to the proof in Theorem \ref{th:1}, a solution of INMO-MF can be written as 
$\boldsymbol{T}_{\mathsf{u}}=\boldsymbol{C}_\mathsf{u}^T\boldsymbol{U}^d(\boldsymbol{S}^d)^{-1}, 
\boldsymbol{T}_{\mathsf{i}}=\boldsymbol{C}_\mathsf{i}^T\boldsymbol{V}^d$, where $\boldsymbol{C}_\mathsf{u}$, $\boldsymbol{C}_\mathsf{i}$ indicate the selected template users/items. Then, 
\begin{small}
  \begin{equation}
    \begin{split}
      |\boldsymbol{Y} - \boldsymbol{E}_{\mathsf{u}}\boldsymbol{E}_{\mathsf{i}}^T|_F
      = &|\boldsymbol{Y} - \boldsymbol{Y}\boldsymbol{C}_\mathsf{i}\boldsymbol{C}_\mathsf{i}^T\boldsymbol{V}^d(\boldsymbol{S}^d)^{-1}
      (\boldsymbol{U}^d)^T\boldsymbol{C}_\mathsf{u}\boldsymbol{C}_\mathsf{u}^T\boldsymbol{Y}|_F \\
      % = &\boldsymbol{Y}\boldsymbol{D}_\mathsf{i}\boldsymbol{V}^d(\boldsymbol{S}^d)^{-1}
      % (\boldsymbol{U}^d)^T\boldsymbol{D}_\mathsf{u}\boldsymbol{Y} \\
      = &|\boldsymbol{Y} - \boldsymbol{Y}(\boldsymbol{E}-\boldsymbol{L}_\mathsf{i})\boldsymbol{V}^d(\boldsymbol{S}^d)^{-1}
      (\boldsymbol{U}^d)^T(\boldsymbol{E}-\boldsymbol{L}_\mathsf{u})\boldsymbol{Y}|_F \\
      = &|\boldsymbol{Y} - \boldsymbol{U}^d\boldsymbol{S}^d(\boldsymbol{V}^d)^T - 
      \boldsymbol{Y}\boldsymbol{L}_\mathsf{i}\boldsymbol{V}^d(\boldsymbol{S}^d)^{-1}(\boldsymbol{U}^d)^T\boldsymbol{L}_\mathsf{u}\boldsymbol{Y} \\
      & + \boldsymbol{U}^d(\boldsymbol{U}^d)^T\boldsymbol{L}_\mathsf{u}\boldsymbol{Y} + \boldsymbol{Y}\boldsymbol{L}_\mathsf{i}\boldsymbol{V}^d(\boldsymbol{V}^d)^T|_F,
    \end{split}
  \end{equation}
\end{small}
where $\boldsymbol{L}_\mathsf{u} \in \{0, 1\}^{n \times n}$ is a diagonal matrix 
and $(\boldsymbol{L}_\mathsf{u})_{i,i}=1$ means that $\mathsf{u}_i$ is not a template user.
For simplicity, we only consider the template users and assume $\boldsymbol{L}_\mathsf{i}=\boldsymbol{0}$. Then,
\begin{equation}
  \begin{split}
    |\boldsymbol{Y} - \boldsymbol{E}_{\mathsf{u}}\boldsymbol{E}_{\mathsf{i}}^T|_F
  = &|\boldsymbol{Y} - \boldsymbol{U}^d\boldsymbol{S}^d(\boldsymbol{V}^d)^T + 
  \boldsymbol{U}^d(\boldsymbol{U}^d)^T\boldsymbol{L}_\mathsf{u}\boldsymbol{Y}|_F \\
  \le &|\boldsymbol{Y} - \boldsymbol{U}^d\boldsymbol{S}^d(\boldsymbol{V}^d)^T|_F +
  |\boldsymbol{U}^d(\boldsymbol{U}^d)^T\boldsymbol{L}_\mathsf{u}\boldsymbol{Y}|_F \\
  = & \epsilon_{min} + |\boldsymbol{U}^d(\boldsymbol{U}^d)^T\boldsymbol{L}_\mathsf{u}\boldsymbol{Y}|_F.
  \end{split}
\end{equation}
To minimize the norm of the additional error matrix $\boldsymbol{U}^d(\boldsymbol{U}^d)^T\boldsymbol{L}_\mathsf{u}\boldsymbol{Y}$, 
% it provides us with some guides to determine the template users. 
% Specifically, we want to minimize $|\boldsymbol{U}^d(\boldsymbol{U}^d)^T\boldsymbol{L}_\mathsf{u}\boldsymbol{Y}|_2^2$.
let $\boldsymbol{U}^d(\boldsymbol{U}^d)^T=(\boldsymbol{s}^{\mathsf{u}}_1, \boldsymbol{s}^{\mathsf{u}}_2, \cdots, \boldsymbol{s}^{\mathsf{u}}_n)$
and $U_{non-tem}$ denotes the set of non-template users. Then, 
\begin{small}
\begin{equation}
\boldsymbol{U}^d(\boldsymbol{U}^d)^T\boldsymbol{L}_\mathsf{u}\boldsymbol{Y} = 
(\sum_{\mathsf{u}_j \in \mathcal{N}_{\mathsf{i}_1}  \cap U_{non-tem}} \boldsymbol{s}^{\mathsf{u}}_j, \cdots, 
\sum_{\mathsf{u}_j \in \mathcal{N}_{\mathsf{i}_m}  \cap U_{non-tem}} \boldsymbol{s}^{\mathsf{u}}_j)
\end{equation}
\begin{equation}\label{eq:metrics} 
  \begin{split}
  |\boldsymbol{U}^d(\boldsymbol{U}^d)^T\boldsymbol{L}_\mathsf{u}\boldsymbol{Y}|_F^2 &= \sum_{\mathsf{i} \in I} 
  |\sum_{\mathsf{u}_j \in \mathcal{N}_{\mathsf{i}}  \cap U_{non-tem}} \boldsymbol{s}^{\mathsf{u}}_j|_2^2 \\
  &\le \sum_{\mathsf{i} \in I}  |\mathcal{N}_{\mathsf{i}}  \cap U_{non-tem}| \sum_{\mathsf{u}_j \in \mathcal{N}_{\mathsf{i}}  
  \cap U_{non-tem}} |\boldsymbol{s}^{\mathsf{u}}_j|_2^2 \\
  &\le \sum_{\mathsf{u}_j \in U_{non-tem}} |\boldsymbol{s}^{\mathsf{u}}_j|_2^2 \sum_{\mathsf{i} \in \mathcal{N}_{\mathsf{u}_j}} |\mathcal{N}_{\mathsf{i}}|
  \end{split}
\end{equation}
\end{small}
Apparently, the above error upper bound can be minimized by selecting the template users with the largest 
$|\boldsymbol{s}^{\mathsf{u}}_j|_2^2 \sum_{\mathsf{i} \in \mathcal{N}_{\mathsf{u}_j}} |\mathcal{N}_{\mathsf{i}}|$.
\end{proof}
Next, we conduct an empirical experiment on a real-world dataset Gowalla, to validate the effectiveness of the theoretical indicator 
$|\boldsymbol{s}^{\mathsf{u}}_j|_2^2 \sum_{\mathsf{i} \in \mathcal{N}_{\mathsf{u}_j}} |\mathcal{N}_{\mathsf{i}}|$ to select the template users ($|\boldsymbol{s}^{\mathsf{i}}_j|_2^2 \sum_{\mathsf{u} \in \mathcal{N}_{\mathsf{i}_j}} |\mathcal{N}_{\mathsf{u}}|$ to select the template items), namely the error-sort indicator. 
We compare our proposed error-sort with other two heuristic indicators, i.e., node degree and page rank \cite{ilprints422}.

\begin{figure}[t]
  \centering
  \includegraphics[width=0.9\linewidth]{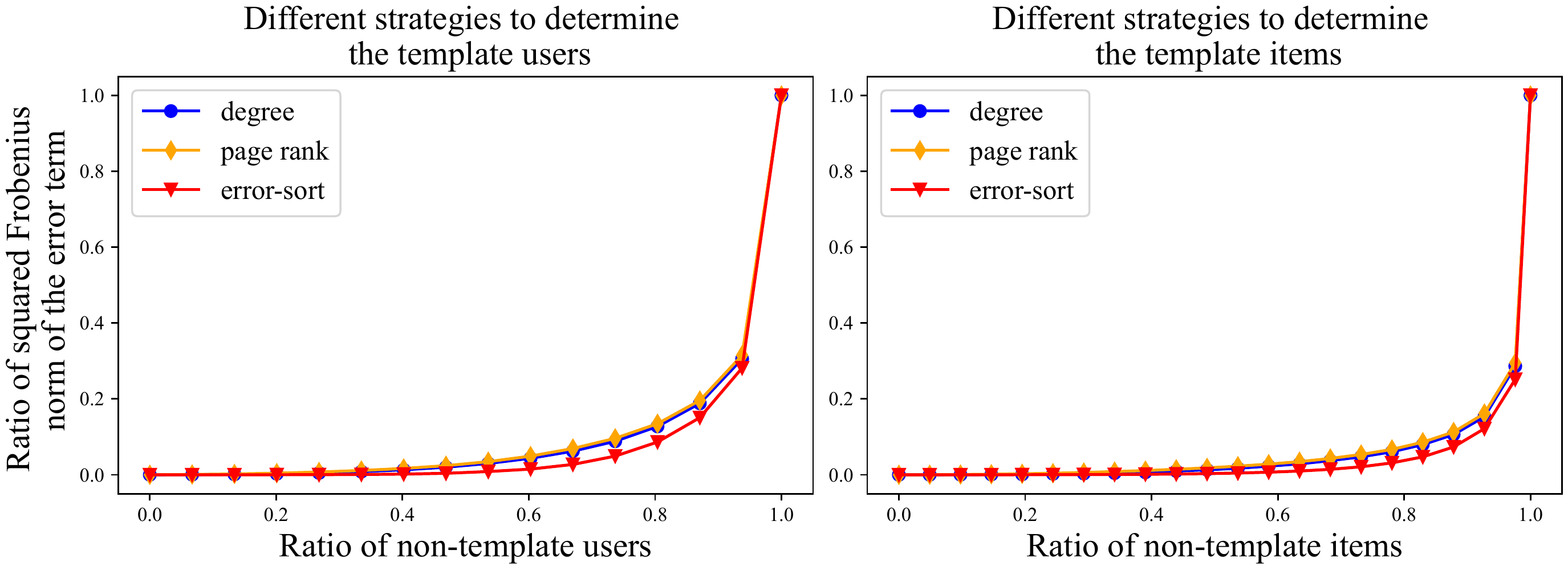}
  \caption{The ratio of squared Frobenius norm of the error term under different ratios of non-template users/items.}
  \Description{Two line charts indicating the trends of the error term by different selection approaches.}
  \label{fig:error}
  %\vspace{-0.3cm}
\end{figure}

Figure \ref{fig:error} demonstrates how the additional error $|\boldsymbol{U}^d(\boldsymbol{U}^d)^T\boldsymbol{L}_\mathsf{u}\boldsymbol{Y}|_F^2$ 
($|\boldsymbol{Y}\boldsymbol{L}_\mathsf{i}\boldsymbol{V}^d(\boldsymbol{V}^d)^T|_F^2$ ) changes with the number of non-template users (non-template items). 
There are two important findings. First, we can ignore $70\%$ users/items and set them as the non-template ones, 
which only leads to an additional error smaller than $10\%$. Besides, we notice the error-sort indicator based on Theorem \ref{th:2} 
is significantly better than another two heuristic indicators.

\subsection{Inductive Embedding Generation}
Based on our above analyses, we now introduce the specific design of INMO.
According to Figure \ref{fig:error}, most of the users/items could be ignored as the non-template ones. 
When obtaining the embeddings of users and items, INMO only utilizes their interactions with the template items $I_{tem}$ and template users $U_{tem}$.
To stabilize the training procedure, we add a denominator to Eq.~(\ref{eq:initialembedding}), adjusting the norms of embeddings.  
Formally, 
\begin{equation} \label{eq:embedding}
 \begin{aligned}
  \boldsymbol{e}_{\mathsf{u}} &= \frac{1}{(|\mathcal{N}_{\mathsf{u}} \cap I_{tem}| + 1)^\alpha} 
  (\sum_{\mathsf{i} \in \mathcal{N}_{\mathsf{u}} \cap I_{tem}}  \boldsymbol{t}_\mathsf{i} + \boldsymbol{t}_{user})
  \quad,\\
  \boldsymbol{e}_{\mathsf{i}} &= \frac{1}{(|\mathcal{N}_{\mathsf{i}}  \cap U_{tem}| + 1)^\alpha} 
  (\sum_{\mathsf{u} \in \mathcal{N}_{\mathsf{i}}  \cap U_{tem}}  \boldsymbol{t}_\mathsf{u} + \boldsymbol{t}_{item})
  \quad,
 \end{aligned}
\end{equation}
\noindent where $\boldsymbol{e}_{\mathsf{u}}, \boldsymbol{e}_{\mathsf{i}}$ denote the inductive embeddings of user $\mathsf{u}$ and item $\mathsf{i}$.
Here, we use $\boldsymbol{t}_{user}$ and $\boldsymbol{t}_{item}$ to model the general characteristics of users and items,
which may help make recommendations for new users and new items with few observed interactions. 
$\alpha$ is an exponent controlling the degree of normalization, which we will discuss in Section \ref{sec:na}.

With the additional denominator, the indicator error-sort should be fine-tuned as 
$|\boldsymbol{s}'^{\mathsf{u}}_j|_2^2 \sum_{\mathsf{i} \in \mathcal{N}_{\mathsf{u}_j}} 1/|\mathcal{N}_{\mathsf{i}}|$, where 
$\boldsymbol{s}'^{\mathsf{u}}_j$ is the $j_{th}$ column of $\boldsymbol{D}_{\mathsf{u}}^{-1}\boldsymbol{U}^d(\boldsymbol{U}^d)^T$
and $\boldsymbol{D}_{\mathsf{u}}$ is the diagonal degree matrix of users. 
In this case, $|\boldsymbol{s}'^{\mathsf{u}}_j|$ will highlight the importance of low-degree users, which may introduce some noises in practice.
Moreover, the calculation of $\boldsymbol{U}^d(\boldsymbol{U}^d)^T$ is expensive. 
Therefore, we implement a simplified version of \textbf{error-sort}, i.e., sorting users by 
$\sum_{\mathsf{i} \in \mathcal{N}_{\mathsf{u}}} 1/|\mathcal{N}_{\mathsf{i}}|$.

Figure \ref{model} shows an example user-item interaction graph, where red circles denote template users and red squares denote template items.
For user $\mathsf{u}_2$, its embedding is aggregated by both the template vector $\boldsymbol{t}_{\mathsf{i}_4}$ of the interacted template item 
$\mathsf{i}_4$ and the global user template $\boldsymbol{t}_{user}$. 
Note that, although item $\mathsf{i}_1$ also interacts with user $\mathsf{u}_2$, according to our error-sort indicator, we neither learn the template vector of $\mathsf{i}_1$ nor use it to represent users.

The proposed INMO has two major advantages: 
1) \textbf{Inductive ability}. When facing new users and new items after training, 
INMO can obtain their embeddings without the need of retraining.
2) \textbf{Adjustable scalability}. The number of parameters in INMO is only dependent on the number of template users 
and template items, which is adjustable according to available computing resources. 

\subsubsection*{Complexity Analysis}
The time complexity of INMO to generate the inductive embedding for user $\mathsf{u}$ is $O(|\mathcal{N}_{\mathsf{u}} \cap I_{tem}|)$, which 
is affordable in most cases. In contrast, IDCF \cite{pmlrwu21j} takes the $O(n)$ complexity to pass messages from key users  
and IGMC \cite{MuhanICLR20} needs to do the graph regression task $O(m)$ times to recommend for a single user.
In general cases, $|\mathcal{N}_{\mathsf{u}} \cap I_{tem}|$ is much smaller than $n$ and $m$, because a user always interacts with a limited number of items. 
The space complexity of INMO is $O((n_t+m_t)d)$, as we only need to save the template user vectors and template item vectors.
The embeddings of other users/items can be generated on the fly, which is memory-efficient.

\subsection{Model Optimization}
\subsubsection{Loss Function}
A commonly adopted training loss for top-k CF methods is the pairwise Bayesian Personalized Ranking (BPR) loss \cite{Rendle2009BPRBP}. 
It encourages the model to predict a higher score for an observed entry than an unobserved entry.
The loss function formulates as follows,
\begin{equation}
  \mathcal{L}_{BPR} = -\sum_{i=1}^{n} \sum_{\mathsf{i}_j \in \mathcal{N}_{\mathsf{u}_i}} \sum_{\mathsf{i}_k \notin \mathcal{N}_{\mathsf{u}_i}} 
  ln\ \sigma (\boldsymbol{r}_{\mathsf{u}_i}^T \boldsymbol{r}_{\mathsf{i}_j} - \boldsymbol{r}_{\mathsf{u}_i}^T \boldsymbol{r}_{\mathsf{i}_k}) + \lambda  ||\Theta||_2^2 \quad,
\end{equation}
\noindent where $\lambda$ controls the $L_2$ regularization strength and $\Theta$ denotes all trainable model parameters.

To facilitate the optimization of template user vectors and template item vectors, we propose an additional self-enhanced loss $\mathcal{L}_{SE}$.
Intuitively, the supervised BPR loss on the final representations $\boldsymbol{r}_\mathsf{u}, \boldsymbol{r}_\mathsf{i}$ may be not enough, 
since the template users/items are aggregated together and lose their identities when generating the inductive embeddings $\boldsymbol{e}_\mathsf{i}, \boldsymbol{e}_\mathsf{u}$.
So we design a new supervised signal to directly guide the learning process of $\boldsymbol{t}_\mathsf{u}, \boldsymbol{t}_\mathsf{i}$ 
for template users and template items, enhancing their identities,

\begin{footnotesize}
\begin{equation}
    \mathcal{L}_{SE} = -\sum_{\mathsf{u}_i \in U_{tem}} \sum_{\mathsf{i}_j \in \mathcal{N}_{\mathsf{u}_i} \cap I_{tem}}  
    \sum_{\mathsf{i}_k \in \overline{\mathcal{N}_{\mathsf{u}_i}} \cap I_{tem}} 
       ln\ \sigma (\boldsymbol{t}_{\mathsf{u}_i}^T \boldsymbol{W}_s \boldsymbol{t}_{\mathsf{i}_j} - 
       \boldsymbol{t}_{\mathsf{u}_i}^T \boldsymbol{W}_s \boldsymbol{t}_{\mathsf{i}_k})~. \\
\end{equation}
\end{footnotesize}

\noindent The final loss is $\mathcal{L} = \mathcal{L}_{BPR} + \beta \mathcal{L}_{SE}$, 
where $\beta$ is a hyper-parameter to balance the BPR loss and the self-enhanced loss.

From another perspective, $\mathcal{L}_{SE}$ is actually the BPR loss from the original MF, but only on the template users and template items. 
Consistent with the analysis in Section \ref{sec:th}, the optimal solution is $\boldsymbol{E}_{\mathsf{u}}=\boldsymbol{U}^d\boldsymbol{S}^d$, 
$\boldsymbol{E}_{\mathsf{i}}=\boldsymbol{V}^d$ for MF 
and $\boldsymbol{T}_{\mathsf{u}}=\boldsymbol{U}^d(\boldsymbol{S}^d)^{-1}$, $\boldsymbol{T}_{\mathsf{i}}=\boldsymbol{V}^d$ for INMO-MF. 
Thus, we add a learnable diagonal matrix $\boldsymbol{W}_s$ to model this difference.

\subsubsection{Normalization Annealing}
\label{sec:na}
Considering the exponent of normalization $\alpha$, we find that constantly setting it to $1$ may overly punish the weights of 
active users with long interaction histories. Especially in the early stage of training, 
hard normalization ($\alpha=1$) may lead to a slow convergence, which is consistent with the findings in \citep{he2018nais}.
Consequently, we adopt an annealing strategy to dynamically control the degree of normalization.
Specifically, in the training phase, $\alpha$ is first set to an initial value of $0.5$ and will be gradually increased to $1$.
In the test phase, $\alpha$ is fixed to $1$.
At the beginning of training, the embeddings of active users would be trained better than non-active users as they have more training data.
Therefore, increasing the norms of active users, that is, making the early normalization exponent $\alpha$ smaller than $1$, 
could temporarily emphasize active users and accelerate the training procedure.

\subsubsection{Drop Interaction}
Since we expect the learned model to be inductive, it is supposed to make recommendations for new users with unseen combinations of past interactions. 
For this reason, we propose to randomly drop the interaction sets $\mathcal{N}_u, \mathcal{N}_{\mathsf{i}}$ with a certain probability
during training as an approach of data augmentation, which also prevents the model from over-fitting.
In this way, INMO could see varying combinations of past interactions from a single user in different training steps, 
improving its inductive ability.

\section{Experiments}
To evaluate the effectiveness of our proposed inductive embedding module, 
we conduct extensive experiments on three public real-world datasets. 
Specifically, the experiments are intended to answer the following three questions:
\begin{itemize}[leftmargin=*]
  \item \textbf{Q1}: How does INMO perform in the transductive scenario as compared with the embedding lookup table?
  \item \textbf{Q2}: How strong is the inductive ability of INMO in inductive scenarios with new interactions and new users/items?
  \item \textbf{Q3}: How do different hyper-parameters (i.e., the strength of the self-enhanced loss and the number of template users/items) 
  and training techniques (i.e., normalization annealing and drop interaction) affect the performance of INMO?
\end{itemize}
Next, we introduce the specific experimental settings and present detailed experimental analyses to each question.

\subsection{Experimental Settings}

\begin{table}[t]
  \caption{Statistics of the datasets.}
  \label{tb:dataset}
  \begin{tabular}{ccccc}
    \toprule
    Dataset & \#Users & \#Items & \#Interactions & Density\\
    \midrule
    Gowalla & 29,858	& 40,988	& 1,027,464 & 0.00084 \\
    Yelp & 75,173 & 42,706 & 1,931,173 & 0.00060 \\
    Amazon-Book & 109,730	& 96,421 & 3,181,759 & 0.00030\\
    \bottomrule
\end{tabular}
%\vspace{-0.5cm}
\end{table}

\begin{table*}
  \caption{Overall performances in the transductive recommendation scenario.}
  \label{tb:trans}
  \scalebox{0.84}{
  \begin{threeparttable}   
  \begin{tabular}{c|ccc|ccc|ccc}
    \toprule
    & \multicolumn{3}{c}{Gowalla} & \multicolumn{3}{c}{Yelp} & \multicolumn{3}{c}{Amazon-book}\\
    \midrule
    & Recall@20 & Precision@20 & NDCG@20 & Recall@20 & Precision@20 & NDCG@20 & Recall@20 & Precision@20 & NDCG@20 \\
   \midrule
   NeuMF & $15.08(\pm 0.07)$	& $3.88(\pm 0.02)$	& $11.20(\pm 0.06)$ & $7.94(\pm 0.05)$	& $1.75(\pm 0.01)$
   & $4.87(\pm 0.04)$	& $\backslash$	& $\backslash$	& $\backslash$	 \\
   \hline
   Mult-VAE & $18.36(\pm 0.08)$	& $4.80(\pm 0.03)$	& $13.90(\pm 0.06)$ & $10.13(\pm 0.04)$	& $2.24(\pm 0.00)$
   & $6.49(\pm 0.02)$	& $13.63(\pm 0.05)$	& $2.81(\pm 0.01)$	& $9.44(\pm 0.05)$	 \\
   \hline
   NGCF & $17.36(\pm 0.08)$	& $4.48(\pm 0.04)$	& $12.77(\pm 0.13)$ & $8.72(\pm 0.22)$	& $1.82(\pm 0.28)$
   & $5.31(\pm 0.37)$	& $12.04(\pm 0.13)$	& $2.42(\pm 0.03)$	& $7.89(\pm 0.08)$	 \\
   \hline
   IMC-GAE & $15.42(\pm 0.09)$	& $3.98(\pm 0.03)$	& $11.63(\pm 0.05)$ & $6.43(\pm 0.05)$	& $1.41(\pm 0.01)$
   & $4.00(\pm 0.04)$	& $8.24(\pm 0.04)$	& $1.76(\pm 0.01)$	& $5.43(\pm 0.03)$	 \\
   \hline
   IDCF-LGCN & $12.80(\pm 0.23)$	& $3.43(\pm 0.05)$	& $9.60(\pm 0.15)$ & $5.49(\pm 0.14)$	& $1.32(\pm 0.02)$
   & $3.55(\pm 0.08)$	& $\backslash$	& $\backslash$	& $\backslash$	 \\
   \hline 
   \hline
   MF-BPR & $16.11(\pm 0.08)$	& $4.14(\pm 0.02)$	& $12.03(\pm 0.07)$ & $8.61(\pm 0.15)$	& $1.90(\pm 0.03)$
   & $5.36(\pm 0.09)$	& $11.41(\pm 0.13)$	& $2.28(\pm 0.03)$	& $7.37(\pm 0.11)$	 \\
   \hline
   INMO-MF & $18.41(\pm 0.10)$	& $4.92(\pm 0.02)$	& $14.15(\pm 0.10)$ & $9.54(\pm 0.03)$	& $2.10(\pm 0.01)$
   & $6.07(\pm 0.03)$	& $13.40(\pm 0.06)$	& $2.81(\pm 0.01)$	& $9.23(\pm 0.04)$	 \\
   \hline
   INMO-MF* & $16.20(\pm 0.09)$	& $4.41(\pm 0.03)$	& $12.40(\pm 0.08)$ & $9.09(\pm 0.09)$	& $2.01(\pm 0.01)$
   & $5.81(\pm 0.05)$	& $11.88(\pm 0.05)$	& $2.53(\pm 0.01)$	& $8.13(\pm 0.04)$	 \\
   \hline
   \hline
   LightGCN & $18.88(\pm 0.13)$	& $4.95(\pm 0.04)$	& $14.18(\pm 0.10)$ & $9.68(\pm 0.08)$	& $2.14(\pm 0.01)$
   & $6.16(\pm 0.06)$	& $12.32(\pm 0.10)$	& $2.56(\pm 0.02)$	& $8.19(\pm 0.06)$	 \\
   \hline
   INMO-LGCN & $ \textbf{20.17}(\pm 0.11)$	& $\textbf{5.36}(\pm 0.04)$	& $\textbf{15.41}(\pm 0.10)$ & $\textbf{10.26}(\pm 0.03)$	& $\textbf{2.25}(\pm 0.01)$
   & $\textbf{6.51}(\pm 0.02)$	& $\textbf{14.28}(\pm 0.05)$	& $\textbf{3.01}(\pm 0.01)$	& $\textbf{9.86}(\pm 0.03)$	 \\
   \hline
   INMO-LGCN* & $19.58(\pm 0.09)$	& $5.21(\pm 0.03)$	& $14.96(\pm 0.06)$ & $10.21(\pm 0.04)$	& $2.23(\pm 0.01)$
   & $6.47(\pm 0.03)$	& $13.77(\pm 0.04)$	& $2.92(\pm 0.01)$	& $9.53(\pm 0.04)$	 \\
   \bottomrule
\end{tabular}
\begin{tablenotes}    
        \footnotesize               
        \item  `$\backslash$': These methods cannot deal with the large dataset Amazon-book.  
      \end{tablenotes}
    \end{threeparttable} }
\end{table*}

\subsubsection{Dataset}
We conduct experiments on three public benchmark datasets:
\textbf{Gowalla}\footnote{\url{https://snap.stanford.edu/data/loc-Gowalla.html}} \cite{cho2011friendship}
, \textbf{Yelp}\footnote{\url{https://www.yelp.com/dataset}}, and \textbf{Amazon-Book}\footnote{\url{http://jmcauley.ucsd.edu/data/amazon/}} \cite{he2016ups}, 
where the items are locations, local businesses, and books respectively. 
For Yelp and Amazon-book, we regard the ratings greater than $3$ as observed interactions and filter out users and items with 
less than $10$ interactions, similar to the pre-processing procedure in \cite{liang2018variational}.
For each user, we randomly split its interactions into $70\%$, $10\%$, and $20\%$ as the train, validation, and test sets.
The experiments are repeated five times with different dataset splits, and the average results with standard deviations are reported.
The characteristics of the three datasets are summarized in Table \ref{tb:dataset}.

\subsubsection{Baseline Methods}
We implement our INMO with the classic MF (INMO-MF) and the state-of-the-art LightGCN (INMO-LGCN) to explore how INMO 
improves the recommendation accuracy and brings the inductive ability. 
We compare our methods with the following baselines:

\begin{itemize}[leftmargin=*]
  \item \textbf{MF-BPR} \cite{koren2009matrix}: It is the most basic latent factor model optimized by BPR loss, yielding
  competitive performances in many cases. 
  \item \textbf{NeuMF} \cite{he2017neural}: It explores the nonlinear interactions between the representations of users and items 
  by a fusion of Hadamard product and MLP. 
  \item \textbf{Mult-VAE} \cite{liang2018variational}: It employs the variational autoencoder architecture~\cite{KingmaW13} to encode 
  and decode users' interaction behaviors. 
  \item \textbf{NGCF} \cite{wang2019neural}: It is a GNN-based recommender model containing the feature transformation 
  and the nonlinear activation like the standard graph convolutional network (GCN) \cite{kipf2016semi}. 
  \item \textbf{LightGCN} \cite{he2020lightgcn}: It is a simplified version of GCN for recommendation, which linearly 
  propagates the representations and achieves the state-of-the-art transductive performance.
  \item \textbf{IMC-GAE} \cite{shen2021inductive}: It employs a postprocessing graph autoencoder for the inductive rating prediction. 
  We adapt it to do the implicit top-k recommendation task.
  \item \textbf{IDCF-LGCN} \cite{pmlrwu21j}: IDCF is a two-stage inductive recommendation framework, estimating the underlying 
  relations from key users to query users. We implement it with LightGCN to recommend for both new users and new items. 
\end{itemize}

\subsubsection{Evaluation Metrics}
We evaluate the effectiveness of CF methods on predicting users' preferences as a ranking problem. 
Specifically, three widely-used evaluation metrics for top-k recommender systems are adopted, i.e., $recall@k$, $precision@k$, and $NDCG@k$. 
We set $k=20$ following \cite{wang2019neural} and report the average values over all users in the test set.
For clarity, we show all the metrics after multiplied by $100$ in the tables of this paper, similar to \cite{he2018nais}.

\begin{table*}
  \caption{Performances in the inductive recommendation scenario with new interactions.}
  \label{tb:newit}
  \scalebox{0.82}{
  \begin{tabular}{c|ccc|ccc|ccc}
    \toprule
    & \multicolumn{3}{c}{Gowalla} & \multicolumn{3}{c}{Yelp} & \multicolumn{3}{c}{Amazon-book}\\
    \midrule
    & Recall@20 & Precision@20 & NDCG@20 & Recall@20 & Precision@20 & NDCG@20 & Recall@20 & Precision@20 & NDCG@20 \\
   \midrule
   INMO-LGCN-retrain & $ 20.17(\pm 0.11)$	& $5.36(\pm 0.04)$	& $15.41(\pm 0.10)$ & $10.26(\pm 0.03)$	& $2.25(\pm 0.01)$
   & $6.51(\pm 0.02)$	& $14.28(\pm 0.05)$	& $3.01(\pm 0.01)$	& $9.86(\pm 0.03)$	 \\
   \hline
   \hline
   Mult-VAE & $16.68(\pm 0.08)$	& $4.40(\pm 0.01)$	& $12.60(\pm 0.05)$ & $9.16(\pm 0.06)$	& $2.06(\pm 0.01)$
   & $5.86(\pm 0.03)$	& $11.79(\pm 0.06)$	& $2.46(\pm 0.01)$	& $8.11(\pm 0.04)$	 \\
   \hline
   Mult-VAE-new & $17.03(\pm 0.06)$	& $4.47(\pm 0.02)$	& $12.89(\pm 0.03)$ & $9.49(\pm 0.04)$	& $\textbf{2.11}(\pm 0.01)$
   & $\textbf{6.07}(\pm 0.02)$	& $12.30(\pm 0.05)$	& $2.55(\pm 0.01)$	& $8.46(\pm 0.02)$	 \\
   \hline
   IMC-GAE & $14.02(\pm 0.05)$	& $3.66(\pm 0.02)$	& $10.61(\pm 0.07)$ & $5.68(\pm 0.04)$	& $1.24(\pm 0.01)$
   & $5.57(\pm 0.03)$	& $6.80(\pm 0.03)$	& $1.50(\pm 0.01)$	& $4.52(\pm 0.03)$	 \\
   \hline
   IMC-GAE-new & $14.39(\pm 0.09)$	& $3.74(\pm 0.03)$	& $10.91(\pm 0.09)$ & $5.85(\pm 0.06)$	& $1.28(\pm 0.02)$
   & $3.68(\pm 0.04)$	& $7.22(\pm 0.02)$	& $1.58(\pm 0.01)$	& $4.80(\pm 0.03)$	 \\
   \hline
   LightGCN & $17.66(\pm 0.24)$	& $4.67(\pm 0.06)$	& $13.41(\pm 0.15)$ & $8.58(\pm 0.07)$	& $1.93(\pm 0.02)$
   & $5.48(\pm 0.06)$	& $10.67(\pm 0.07)$	& $2.27(\pm 0.02)$	& $7.14(\pm 0.06)$	 \\
   \hline
   LightGCN-new & $17.95(\pm 0.24)$	& $4.73(\pm 0.06)$	& $13.53(\pm 0.16)$ & $8.96(\pm 0.08)$	& $1.98(\pm 0.02)$
   & $5.69(\pm 0.06)$	& $11.65(\pm 0.08)$	& $2.46(\pm 0.02)$	& $7.80(\pm 0.07)$	 \\
   \hline
   IDCF-LGCN & $12.54(\pm 0.21)$	& $3.39(\pm 0.07)$	& $9.45(\pm 0.17)$ & $5.16(\pm 0.06)$	& $1.26(\pm 0.01)$
   & $3.34(\pm 0.03)$	& $\backslash$	& $\backslash$	& $\backslash$	 \\
   \hline
   IDCF-LGCN-new & $12.93(\pm 0.29)$	& $3.46(\pm 0.08)$	& $9.69(\pm 0.20)$ & $5.43(\pm 0.08)$	& $1.31(\pm 0.02)$
   & $3.52(\pm 0.05)$	& $\backslash$	& $\backslash$	& $\backslash$	 \\
   \hline
   \hline
   INMO-MF & $16.66(\pm 0.08)$	& $4.49(\pm 0.03)$	& $12.81(\pm 0.12)$ & $8.45(\pm 0.05)$	& $1.89(\pm 0.01)$
   & $5.37(\pm 0.03)$	& $11.30(\pm 0.04)$	& $2.42(\pm 0.00)$	& $7.76(\pm 0.01)$	 \\
   \hline
   INMO-MF-new & $17.45(\pm 0.08)$	& $4.65(\pm 0.04)$	& $13.37(\pm 0.13)$ & $8.98(\pm 0.06)$	& $1.97(\pm 0.01)$
   & $5.70(\pm 0.04)$	& $12.21(\pm 0.06)$	& $2.57(\pm 0.01)$	& $8.32(\pm 0.07)$	 \\
   \hline
   INMO-LGCN & $18.25(\pm 0.05)$	& $4.92(\pm 0.02)$	& $14.03(\pm 0.06)$ & $9.23(\pm 0.06)$	& $2.05(\pm 0.01)$
   & $5.87(\pm 0.04)$	& $12.16(\pm 0.05)$	& $2.64(\pm 0.01)$	& $8.39(\pm 0.01)$	 \\
   \hline
   INMO-LGCN-new & $\textbf{19.21}(\pm 0.04)$	& $\textbf{5.05}(\pm 0.02)$	& $\textbf{14.51}(\pm 0.03)$ & $9.58(\pm 0.09)$	& $2.07(\pm 0.01)$
   & $6.03(\pm 0.05)$	& $\textbf{13.44}(\pm 0.04)$	& $\textbf{2.79}(\pm 0.00)$	& $\textbf{9.02}(\pm 0.03)$	 \\
   \hline
   INMO-LGCN*-new & $18.61(\pm 0.19)$	& $4.92(\pm 0.05)$	& $14.13(\pm 0.15)$ & $\textbf{9.63}(\pm 0.07)$	& $2.07(\pm 0.01)$
   & $6.06(\pm 0.05)$	& $13.01(\pm 0.07)$	& $2.71(\pm 0.02)$	& $8.76(\pm 0.06)$	 \\
   \bottomrule
\end{tabular}}
%\vspace{-0.1cm}
\end{table*}

\subsubsection{Parameter Settings}
We implement our INMO-MF, INMO-LGCN, and all other baseline methods based on Pytorch. 
The codes including dataset processing, hyper-parameter tuning, and model implementations are accessible here
\footnote{\url{https://github.com/WuYunfan/igcn_cf}}.
All models are learned via optimizing the BPR loss, except that NeuMF uses the
binary cross-entropy loss and Mult-VAE maximizes the multinomial likelihood as proposed in their
original papers.
We use the Adam optimizer \cite{DieICLR15} to train all the models for at most $1000$ epochs.
The embedding size of different models is fixed to $64$ for a fair comparison, and the batch size is fixed to $2048$.
We apply a grid search on the validation set, tuning hyper-parameters for INMO and other baseline methods,
the learning rate is tuned over $\{10^{-4}, 10^{-3}, 10^{-2}\}$, the $L_2$ regularization coefficient over $\{0, 10^{-5}, 10^{-4}, 10^{-3}, 10^{-2}\}$, 
and the dropout rate over $\{0, 0.1, 0.3, 0.5, 0.7, 0.9\}$.
We set the number of graph convolution layers to three for all graph-based methods.
The sampling size of IDCF-LGCN is set to $50$ for an affordable time consumption.
Moreover, the early stopping strategy is performed, i.e., stopping training if $NDCG@20$ on the validation data does not increase for $50$ successive epochs.
To comprehensively demonstrate the effectiveness of our method, 
we implement two specific versions of INMO, denoted as INMO-MF, INMO-LGCN and INMO-MF*, INMO-LGCN*. 
Specifically, INMO-MF and INMO-LGCN take all users/items as the template users/items, 
while INMO-MF* and INMO-LGCN* take only $30\%$ of users as the template users and $30\%$ of items as the template items.
By default, the template users and template items are selected by the error-sort indicator.
A detailed analysis about the number of template users and template items is provided in Section \ref{sec:tem}.

\subsection{Transductive Recommendation (Q1)}
We first conduct traditional transductive experiments to demonstrate the general effectiveness of our proposed INMO.
Table \ref{tb:trans} shows the recommendation results of all methods in the transductive scenario
\footnote{As we run $5$ times on different divisions of datasets, 
our results are a little different from those reported in \cite{he2020lightgcn}.}. 
The results include seven baseline methods, and four models using our INMO, including INMO-MF, INMO-MF*, INMO-LGCN, and INMO-LGCN*. 
We have the following observations:
\begin{itemize}[leftmargin=*]
  \item Our proposed INMO-LGCN model outperforms all other baseline methods in the transductive recommendation scenario, 
  beating state-of-the-art recommenders Mult-VAE and LightGCN.
  \item Both of our implemented INMO-MF and INMO-LGCN significantly outperform their basic versions MF and LightGCN, which employ a traditional embedding lookup table. 
  It suggests the superiority of INMO in the transductive recommendation scenario, which can generate more accurate recommendation results.
  \item INMO-MF* and INMO-LGCN* show better performances than MF and LightGCN, while with only $30\%$ parameters, 
  indicating the potential of our INMO in resource limited applications.
\end{itemize}
In INMO-MF and INMO-LGCN, two users with the same historical behaviors will obtain the same embedding, i.e., the embedding mapping function is injective.
While in their original versions (MF and LightGCN), these two users may have different recommended items, owing to their randomly initialized individual embeddings.
Such injective property may further help the recommenders to make the most of the training data and reduce noises, leading to better performances.

\subsection{Inductive Recommendation (Q2)}
A great advantage of our INMO lies in its capability to model new interactions and new users/items in the test phase without the need of retraining.
Thus, we conduct experiments in two inductive recommendation scenarios with \emph{new interactions} and \emph{new users/items}, 
as described in Section \ref{sec:settings}.

\begin{table*}[ht!]
  \caption{Performances in the inductive recommendation scenario with new users and new items.}
  \label{tb:newui}
  \scalebox{0.86}{
  \begin{threeparttable}
  \begin{tabular}{c|ccc|ccc|cccc}
    \toprule
    & \multicolumn{3}{c}{Gowalla} & \multicolumn{3}{c}{Yelp} & \multicolumn{3}{c}{Amazon-book}\\
      \midrule
    & New User & New Item	& Over All & New User & New Item	& Over All & New User & New Item	& Over All \\
   \midrule
   INMO-LGCN-retrain & $14.01(\pm 0.42)$	& $16.20(\pm 0.20)$	& $15.41(\pm 0.10)$ & $6.21(\pm 0.07)$	& $13.21(\pm 0.10)$
   & $6.51(\pm 0.02)$	& $14.73(\pm 0.12)$	& $14.91(\pm 0.19)$	& $9.86(\pm 0.03)$	 \\
   \hline
   \hline
   Popular & $1.54(\pm 0.12)$	& $0.91(\pm 0.06)$	& $2.10(\pm 0.03)$ & $1.04(\pm 0.03)$	& $1.91(\pm 0.05)$
   & $1.01(\pm 0.02)$	& $0.55(\pm 0.02)$	& $0.82(\pm 0.03)$	& $0.70(\pm 0.01)$	 \\
   \hline
   Mult-VAE & $10.77(\pm 0.31)$	& $\backslash$	& $12.58(\pm 0.08)$ & $4.93(\pm 0.07)$	& $\backslash$
   & $5.70(\pm 0.05)$	& $\textbf{9.12}(\pm 0.06)$	& $\backslash$	& $7.87(\pm 0.04)$	 \\
   \hline
   IMC-GAE & $8.42(\pm 0.15)$	& $9.25(\pm 0.21)$	& $9.81(\pm 0.08)$ & $2.57(\pm 0.03)$	& $7.54(\pm 0.11)$
   & $3.25(\pm 0.03)$	& $4.91(\pm 0.35)$	& $5.12(\pm 0.15)$	& $4.19(\pm 0.05)$	 \\
   \hline
   IMC-LGCN & $10.38(\pm 0.31)$	& $10.90(\pm 0.09)$	& $13.24(\pm 0.12)$ & $4.67(\pm 0.04)$	& $10.74(\pm 0.08)$
   & $5.57(\pm 0.03)$	& $7.18(\pm 0.19)$	& $7.07(\pm 0.16)$	& $6.39(\pm 0.07)$	 \\
   \hline
   IDCF-LGCN & $8.29(\pm 0.18)$	& $8.60(\pm 0.13)$	& $9.40(\pm 0.09)$ & $3.28(\pm 0.04)$	& $7.77(\pm 0.10)$
   & $3.48(\pm 0.06)$	& $\backslash$	& $\backslash$	& $\backslash$	 \\
   \hline
   \hline
   INMO-MF & $10.85(\pm 0.24)$	& $10.92(\pm 0.14)$	& $13.10(\pm 0.06)$ & $4.85(\pm 0.24)$	& $10.73(\pm 0.25)$
   & $5.50(\pm 0.05)$	& $1.89(\pm 0.66)$	& $0.63(\pm 0.26)$	& $1.56(\pm 0.59)$	 \\
   \hline
   INMO-LGCN & $\textbf{12.36}(\pm 0.38)$	& $\textbf{13.62}(\pm 0.08)$	& $\textbf{14.52}(\pm 0.11)$ & $\textbf{5.75}(\pm 0.08)$	& $\textbf{12.17}(\pm 0.05)$
   & $\textbf{6.13}(\pm 0.02)$	& $9.05(\pm 0.05)$	& $\textbf{7.99}(\pm 0.27)$	& $\textbf{7.94}(\pm 0.07)$	 \\
   \hline
   INMO-LGCN* & $10.95(\pm 0.28)$	& $11.07(\pm 0.10)$	& $13.49(\pm 0.05)$ & $5.56(\pm 0.08)$	& $12.16(\pm 0.07)$
   & $6.10(\pm 0.02)$	& $7.29(\pm 0.16)$	& $6.78(\pm 0.07)$	& $7.20(\pm 0.03)$	 \\
   \bottomrule
  \end{tabular}
  \begin{tablenotes}    
        \footnotesize               
        \item `$\backslash$': Mult-VAE cannot handle the inductive scenario with new items and IDCF cannot apply to large datasets.
      \end{tablenotes}         
    \end{threeparttable}} 
\end{table*}

\subsubsection{New Interactions}
In this scenario, we randomly remove $20\%$ of training interactions from each user, then train the recommender models on the remaining data. 
During the test phase, previously removed interactions arrive as the new interactions, 
which can be further utilized to improve the recommendation performances.

Note that not all of the baseline methods can handle this inductive scenario. 
MF-based methods (MF and NeuMF) cannot make recommendations with new interactions or new users/items. 
We adapt Mult-VAE to this scenario by adding the new interactions to its updated inputs of the encoder.
GNN-based CF methods, i.e., NGCF, LightGCN, and IMC-GAE, take the new interactions into consideration via adding new links in the interaction graph.
As for our INMO, it updates the inductive embeddings of users and items, enhancing the utilization of new interactions.

We present the experimental results in Table \ref{tb:newit}.
The suffix \textbf{-new} indicates the updated results considering additional new interactions in the test phase, otherwise not.
\textbf{INMO-LGCN-retrain} refers to the performances of INMO-LGCN through retraining to incorporate the new interactions, 
served as the performance upper bound.
As shown in Table \ref{tb:newit}, new interactions help to improve the performances for all methods, 
verifying the benefits of modeling additional new interactions in the test phase.
Both INMO-MF and INMO-LGCN significantly outperform their basic versions on all datasets, 
whether adding new interactions or not. 
Especially after adding new interactions, INMO-LGCN-new increases $NDCG@20$ by $7.24\%$, $5.98\%$, $15.64\%$ 
on Gowalla, Yelp, and Amazon-book, compared with LightGCN-new.
The results empirically validate the inductive capability of our INMO by considering the new interactions
in the embedding generation.
% Note that, INMO-LGCN-new approaches or even outperforms INMO-LGCN-Retrain, i.e., the upper bound in the transductive scenario, 
% indicating that our proposed method obtains satisfactory performances in the inductive scenario,  
% comparable to the model performances after retraining while saving a lot of computing resources.

\subsubsection{New Users/Items}
In this scenario, we randomly remove $20\%$ of users and items from the training and validation data,
but the test set keeps the same as the one used in the transductive scenario. This is a common scenario in real-world services, 
which means the methods need to recommend both new items and old items for new users who have not been seen during training.

Since NGCF and LightGCN need to learn the user-specific and item-specific embeddings, 
they can not work on the inductive recommendation scenario with new users and new items.
To better demonstrate the effectiveness of INMO, we adapt LightGCN to this scenario by 
employing the same postprocessing strategy in IMC-GAE \cite{shen2021inductive}, denoted as IMC-LGCN.
As for Mult-VAE, which is intrinsically unable to recommend new items without retraining, 
we evaluate their performances when giving only recommendations with old items. 
In addition, we introduce a non-personalized recommender \textbf{Popular}, which recommends the most popular items for all users, 
as the lower bound of performances.

Table \ref{tb:newui} shows the performance comparison in the new users and new items scenario in terms of $NCDG@20$.
In addition to the overall performances, we report the average $NCDG@20$ of new users and the retrieval results among new items.
INMO-LGCN achieves the best recommendation results among inductive methods, significantly outperforming all the baseline models, 
approaching the upper bound of INMO-LGCN-retrain.
It indicates the proposed INMO is quite effective in generalizing to new users and new items which are unseen during training.         

\begin{figure}[t]
  \centering
  \includegraphics[width=0.99\linewidth]{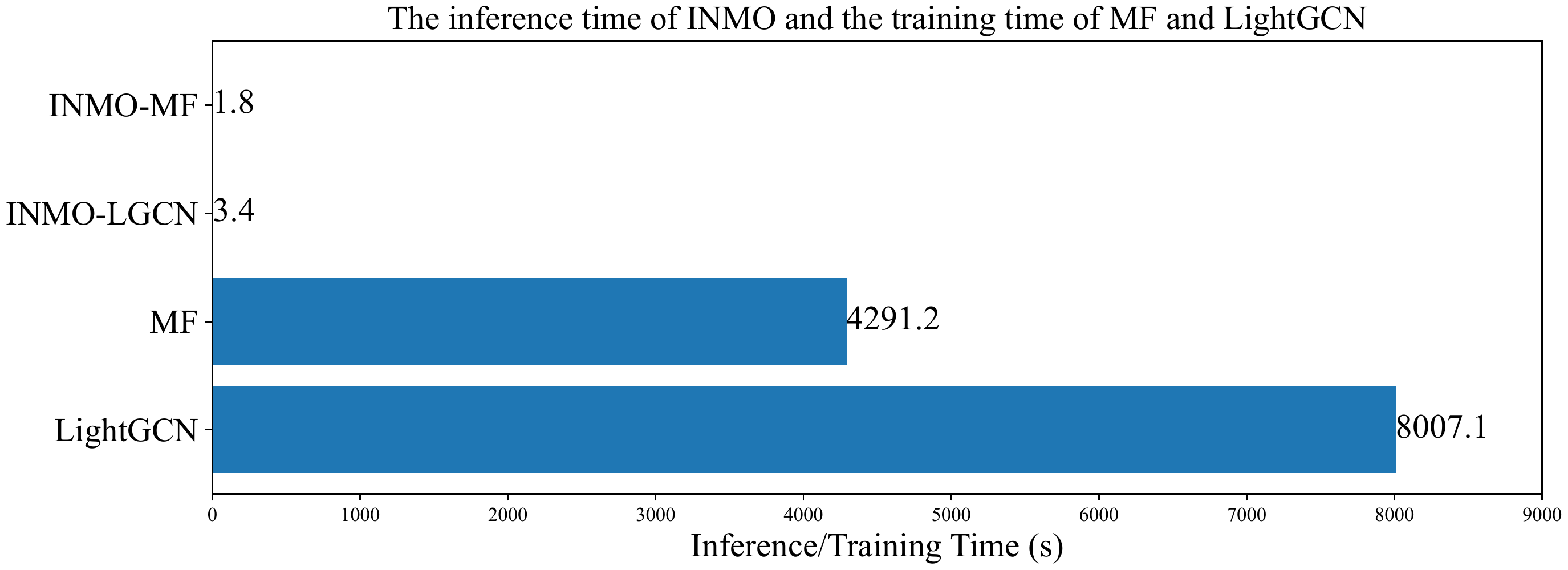}
  \caption{The retraining cost of two representative latent factor models.}
  \Description{A histogram indicating the inference time of INMO and the training time of MF and LightGCN}
  \vspace{-0.2cm}
  \label{fig:time}
\end{figure}

\subsubsection*{Retraining Cost}
INMO can avoid the frequent retraining of recommender models, which is of great value in real-world applications. 
Specifically, the full retraining of the LightGCN model with a Tesla V100 in a small dataset Gowalla still takes more than $2$ hours, while our INMO is
able to inductively recommend for new users and new items with an inference time of only several seconds.
Figure~\ref{fig:time} illustrates the time consumptions with and without INMO when facing new users and new items. 
It is evident that our proposed inductive embedding module can save a lot of computing resources and provide timely and accurate recommendations.

\subsection{Hyper-parameter Analysis (Q3)}
In this section, we conduct experiments to analyze the impact of some hyper-parameters and training techniques. 
%We take the Gowalla dataset as an example for analysis, due to the space limitation,

\begin{figure}[t]
  \centering
  \includegraphics[width=\linewidth]{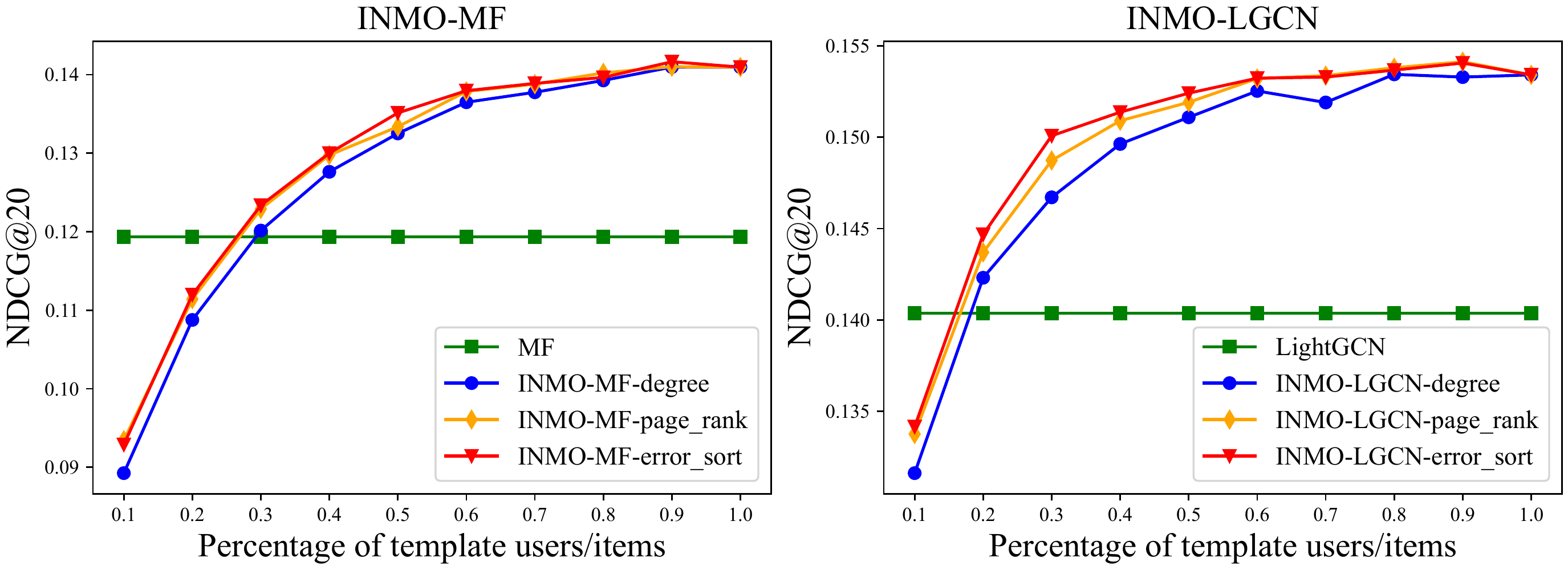}
  \caption{The recommendation performances under different percentages of template users and template items.}
  \Description{Two line charts indicating the influence of varying the percentage of template users and template items.}
%   \vspace{-0.5cm}
  \label{core}
\end{figure}

\subsubsection{The Number of Template Users and Template Items}
\label{sec:tem}
We explore the influence of template users and template items on the recommendation performances, 
and compare various ways to select the template users/items as mentioned in Section~\ref{sec:th}. 
Experiments are conducted for both INMO-MF and INMO-LGCN on Gowalla dataset (shown in Figure~\ref{core}).
It is observed that both INMO-MF and INMO-LGCN can yield better performances than their original versions, while with much fewer parameters. 
Specifically, INMO-MF outperforms MF with only $30\%$ model parameters, and INMO-LGCN beats the state-of-the-art LightGCN with 
only $20\%$ parameters. 
Figure~\ref{core} also empirically validates that the error-sort indicator could guide to select a better set of template users/items, leading to a higher recommendation accuracy.  
These findings demonstrate the effectiveness and scalability of our proposed INMO.

\begin{figure}[t]
  \centering
  \includegraphics[width=0.9\linewidth]{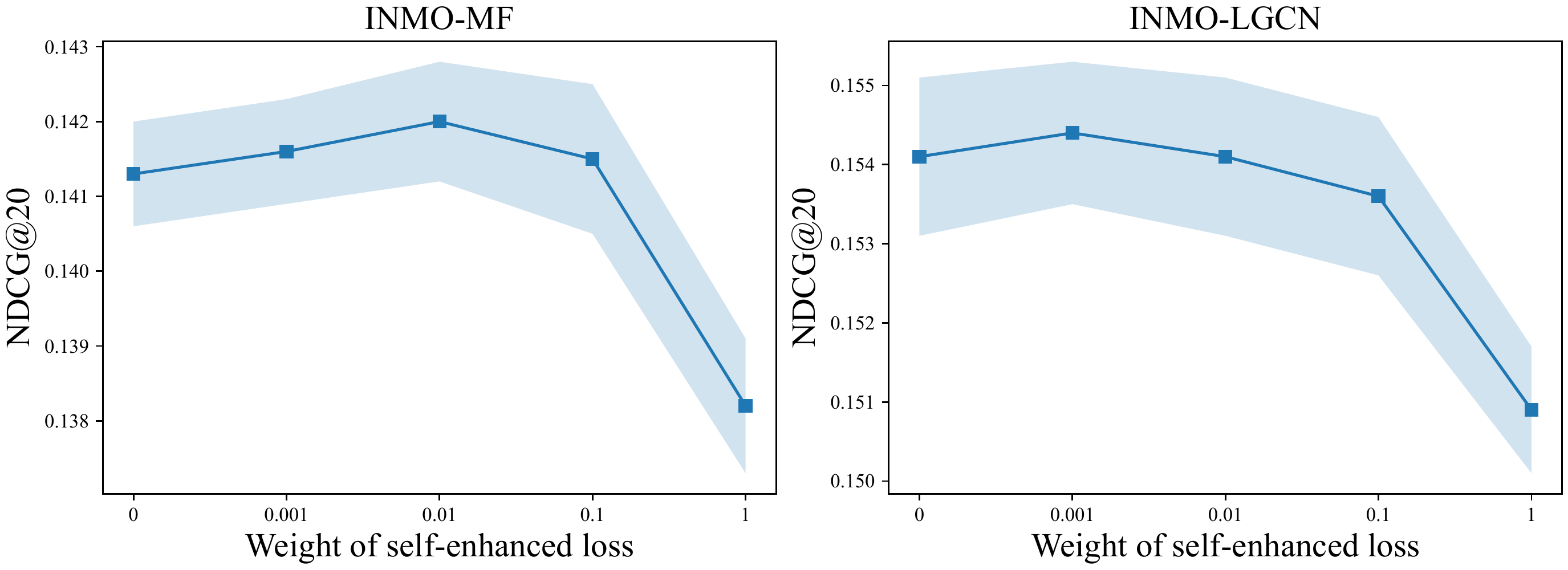}
  \caption{The recommendation performances under different weights of the self-enhanced loss $\beta$.}
  \Description{Two line charts indicating the influence of varing the strength of the self-enhanced loss.}
  \vspace{-0.2cm}
  \label{loss}
\end{figure}

\subsubsection{Self-enhanced loss}
To demonstrate the impact of our proposed self-enhanced loss $\mathcal{L}_{SE}$, we evaluate the recommendation performances 
under different weights of the self-enhanced loss on the Gowalla dataset.
As shown in Figure~\ref{loss}, we find that INMO-MF yields the best performance with a $\beta$ at $0.1$, 
while INMO-LGCN at $0.01$, indicating that INMO-MF needs more additional supervised information. 
These results prove the general effectiveness of $\mathcal{L}_{SE}$ on various backbone models and 
suggest that the strength of the self-enhanced loss should be carefully tuned in different situations. 

\subsubsection{Ablation Study}
\begin{table}[t]
  \caption{The ablation study on two training techniques.}
  \label{tb:ablation}
  \scalebox{0.78}{
  \begin{tabular}{cccccc}
    \toprule
    Method & NA & DI & New User & New Item & Over All  \\
    \midrule
    INMO-MF & $\checkmark$ & $\checkmark$ & $\textbf{10.85}(\pm 0.24)$	& $\textbf{10.92}(\pm 0.14)$	& $\textbf{13.10}(\pm 0.06)$ \\
    \hline
     & $\times$ & $\checkmark$ & $3.29(\pm 0.27)$	& $4.06(\pm 0.29)$	& $5.31(\pm 0.21)$ \\
    \hline
     & $\checkmark$ & $\times$ & $10.45(\pm 0.23)$	& $10.70(\pm 0.12)$	& $12.83(\pm 0.06)$ \\
    \hline
     & $\times$ & $\times$ & $6.44(\pm 0.41)$	& $7.09(\pm 0.38)$	& $9.29(\pm 0.26)$ \\
    \hline
    \hline
    INMO-LGCN & $\checkmark$ & $\checkmark$ & $12.36(\pm 0.38)$	& $\textbf{13.62}(\pm 0.08)$	& $\textbf{14.52}(\pm 0.11)$ \\
    \hline
     & $\times$ & $\checkmark$ & $12.26(\pm 0.32)$	& $13.42(\pm 0.09)$	& $14.37(\pm 0.09)$ \\
    \hline
     & $\checkmark$ & $\times$ & $\textbf{12.38}(\pm 0.39)$	& $13.44(\pm 0.11)$	& $\textbf{14.52}(\pm 0.07)$ \\
    \hline
     & $\times$ & $\times$ & $12.17(\pm 0.35)$	& $13.05(\pm 0.15)$	& $14.26(\pm 0.11)$ \\
    \bottomrule
\end{tabular}}
\end{table}

\begin{figure}[t]
 \setlength{\abovecaptionskip}{-0.002cm}
  \centering
  \includegraphics[width=0.9\linewidth]{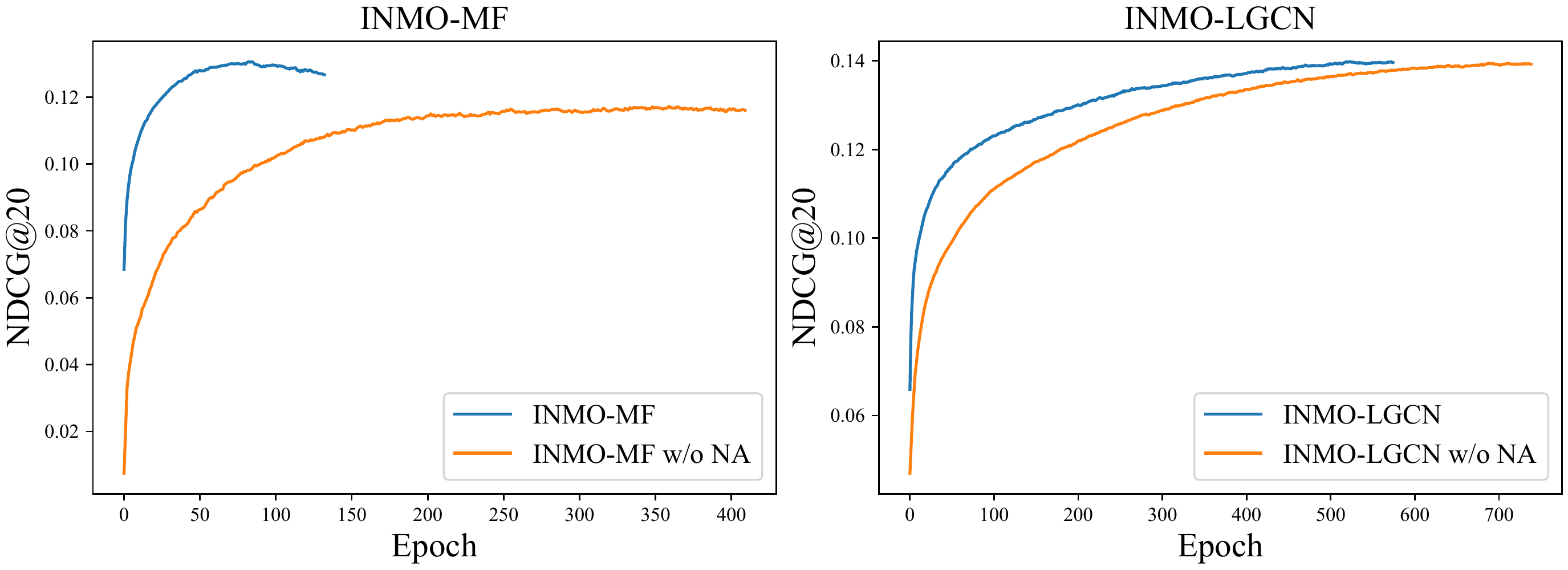}
  \caption{The training procedures with and without the normalization annealing technique.}
  \Description{Two line charts indicating the training procedures with and without normalization annealing.}
  \vspace{-0.2cm}
  \label{fig:na}
\end{figure}

We conduct experiments to evaluate the effectiveness of the training techniques proposed for INMO.
We ablate the normalization annealing (NA) and drop interaction (DI) techniques in INMO-MF and INMO-LGCN, and then evaluate them in the inductive scenario with new users and new items on Gowalla dataset. 
The comprehensive experimental results are presented in Table~\ref{tb:ablation}, which empirically verifies the effectiveness of both the two training techniques, especially for INMO-MF. 
It illustrates that, the optimization procedure for INMO-MF can be significantly improved with our well designed training techniques. It is necessary for INMO to adopt a dynamic normalization strategy and see varying combinations of interactions during training. Specifically, these techniques help INMO-MF  to increase the $NDCG@20$ by $68.48\%$ and $54.02\%$ for new users and new items respectively. 
We realize that the hyper-parameter settings are significantly important for INMO-MF to achieve a competent inductive recommendation performance. In the case of INMO-MF without NA, it yields a passable transductive accuracy, while performing poorly when facing new users and new items.

To further investigate how normalization annealing accelerates the model training, we delve into the training procedures of INMO-MF and IMNO-LGCN with and without this technique.
The $NDCG@20$ on the validation set of Gowalla during training is illustrated in Figure~\ref{fig:na}.
The lines end at different epochs as a result of the early stopping strategy.
We can notice that the variants with normalization annealing converge much faster to the plateau and even achieve better performances, demonstrating the effectiveness of normalization annealing for better model optimization.

\section{Conclusion}
In this work, we propose a novel Inductive Embedding Module, namely INMO, to make recommendations in the inductive scenarios with 
\emph{new interactions} and \emph{new users/items} for collaborative filtering.
INMO generates the inductive embeddings for users and items by considering their past interactions with some template users and template items. 
Remarkably, INMO is model-agnostic and scalable, which is applicable to existing latent factor models and has an adjustable number of parameters.
To demonstrate the effectiveness and generality of our proposed INMO, we attach it to MF and LightGCN and obtain the inductive variants INMO-MF and INMO-LGCN. 
We evaluate INMO on three public real-world benchmarks across both transductive and inductive recommendation scenarios.
Experimental results demonstrate that, INMO-MF* and INMO-LGCN* outperform their original versions even with only $30\%$ of parameters.
Furthermore, INMO-LGCN yields the best performances in all the scenarios.
We hope this work provides some new ideas for researchers to consider the inductive recommendation task, which is a common scenario in real-world services.
\begin{acks}
This work is funded by the National Natural Science Foundation of China under Grant Nos. 62102402, U21B2046, and the National Key R\&D Program of China (2020AAA0105200). Huawei Shen is also supported by Beijing Academy of Artificial Intelligence (BAAI).
\end{acks}
\bibliographystyle{ACM-Reference-Format}
\bibliography{sample-base}


\end{document}